\documentclass[12pt]{iopart}
\usepackage{iopams}  
\expandafter\let\csname equation*\endcsname\relax
\expandafter\let\csname endequation*\endcsname\relax
\usepackage{amsmath}
\usepackage{amssymb}
\usepackage{amscd}
\usepackage{graphicx}
\usepackage[T1]{fontenc}
\usepackage{color}
\usepackage{colordvi}
\usepackage{xspace}
\usepackage{amsthm}
\usepackage[english]{babel}

\usepackage{xhfill}

\usepackage[utf8]{inputenc}

\usepackage[bookmarks=false,pdfstartview={FitH}]{hyperref}

\usepackage{bbm}
\newcommand{\unity}{\mathbbmss{1}}

\newcommand{\N}{\ensuremath{\mathbb N}}
\newcommand{\Z}{\ensuremath{\mathbb Z}}
\newcommand{\R}{\ensuremath{\mathbb R}}
\newcommand{\C}{\ensuremath{\mathbb C}}

\newcommand{\idty}{\mathbbmss{1}}
\renewcommand{\Bbb}[1]{\if1#1\idty\else\mathbb{#1}\fi}

\newcommand{\kb}[1]{|#1\rangle\langle#1|}
\newcommand{\KB}[2]{|#1\rangle\langle#2|}
\newcommand{\ket}[1]{|#1\rangle}
\newcommand{\SP}{\operatorname{span}}
\newcommand{\uu}{\mathfrak{u}}
\newcommand{\su}{\mathfrak{su}}
\newcommand{\fg}{\mathfrak{g}}

\newtheorem{thm}{Theorem}[section]
\newtheorem{defi}[thm]{Definition}
\newtheorem{prop}[thm]{Proposition}
\newtheorem{lem}[thm]{Lemma}
\newtheorem{kor}[thm]{Corollary}

\newtheorem{aX}{Axiom}

\newcommand{\FPV}{D_X}

\newcommand{\BLUE}[1]{\textcolor{blue}{#1}}

\newcommand{\sym}{\operatorname{sym}}
\newcommand{\Sym}{\operatorname{Sym}}

\makeatletter
\renewcommand\tableofcontents{%
  \section*{\contentsname}
    \@starttoc{toc}%
    }
\makeatother

\begin{document}
\selectlanguage{english}

\title[Controlling Several Atoms in a Cavity]{Controlling Several Atoms in a Cavity}

\author{Michael Keyl}

\address{Zentrum Mathematik, M5,
Technische Universit{\"a}t M{\"u}nchen,\\
Boltzmannstrasse 3, 85748 Garching, Germany}
\ead{michael.keyl@tum.de}

\author{Robert Zeier}

\address{Department Chemie,
Technische Universit{\"a}t M{\"u}nchen,\\
Lichtenbergstrasse 4, 85747 Garching, Germany}
\ead{robert.zeier@ch.tum.de}

\author{Thomas {Schulte-Herbr{\"u}ggen}}

\address{Department Chemie,
Technische Universit{\"a}t M{\"u}nchen,\\
Lichtenbergstrasse 4, 85747 Garching, Germany}
\ead{tosh@ch.tum.de}


\vspace{3mm}
\hspace{17mm}{\footnotesize (Dated: January 22, 2014)}\\
\vspace{-3mm}


\begin{abstract}
We treat control  of several two-level atoms interacting
with one mode of the electromagnetic field in a cavity.  
This provides a useful model to study pertinent
aspects of quantum control in infinite dimensions via 
the emergence of infinite-dimensional system algebras. 
Hence we address problems arising with infinite-dimensional Lie algebras and those of unbounded
operators. For the models considered, these problems can be
solved by splitting the set of control Hamiltonians into two subsets:
The first obeys an abelian symmetry and can be treated in terms
of infinite-dimensional Lie algebras and strongly closed subgroups of
the unitary group of the system Hilbert space. The second breaks this
symmetry, and its discussion introduces new arguments. Yet, full
controllability can be achieved in a strong sense: e.g., in a time dependent 
Jaynes-Cummings model we
show that, by tuning coupling constants appropriately, every
unitary of the coupled system (atoms and cavity) can be approximated
with arbitrarily small error.
\end{abstract}

\maketitle

\BLUE{\scriptsize
\tableofcontents
}

\section{Introduction}

Exploiting controlled dynamics of quantum systems is becoming of increasing importance not only 
for solving computational tasks or quantum-secured communication, but also for 
simulating other physical  systems \cite{Fey82, VC02, WRJB02, ZGB, JC03, Lewenstein12}. An interesting
direction in quantum simulation applies many-body correlations to create \/`quantum matter\/'. 
E.g., ultra-cold atoms in optical lattices are versatile models for
studying large-scale correlations \cite{BDZ08, Lewenstein12}. 
Tunability and control over the system parameters of optical lattices allows for switching between 
several low-energy states of different quantum phases \cite{Sachdev99, QPT10} or in particular for following real-time 
dynamics such as the quantum quench from the super-fluid to the Mott insulator regime \cite{GMEHB02}.

Thus manipulating several atoms in a cavity is a key step to this end \cite{Haroche06} at the same time
posing challenging infinite-dimensional control problems.
While in finite dimensions controllability can readily be assessed by the Lie-algebra
rank condition \cite{SJ72, JS72, Bro72, Bro73, Jurdjevic97}, infinite-dimensional systems are
more intricate \cite{LY95}. As exact controllability in infinite dimensions seemed daunting in earlier work \cite{HTC83,TR03,LTCC05,WTL06}, 
it took a while before approximate control paved the way to more realistic assessment \cite{AB05,CMS+09,BGR+13},
for a recent (partial) review see, e.g., also \cite{Borzi2011rev} and references therein. 

Here we explore systems and control aspects for systems consisting of several two-level atoms coupled to a 
cavity mode, i.e. the {\em Jaynes-Cummings model} \cite{JC63, TC68, TC69, BRX98}. 
We build upon our previous symmetry arguments \cite{ZS11,ZZKS14} and moreover, we 
apply appropriate operator topologies for addressing two controllability problems in particular:
(i) to which extent can pure states be interconverted and 
(ii) can unitary gates be approximated with arbitrary precision.
In particular by treating the latter, we go beyond previous work,
which started out by a finite-dimensional truncation of a
two-level atom coupled to an oscillator \cite{RBM+04} followed by generalisations to infinite 
dimensions \cite{BRB03,YL07,BBR10} both being confined to establishing
criteria of pure-state controllability. Note that \cite{BBR10} also treats one atom coupled to several
oscillators.

The general aim of this paper is twofold: On the one hand we study control
problems for atoms interacting with electromagnetic fields in cavities. On the
other hand, we address quantum control in infinite dimensions. 
Therefore, the
purpose of  Section~\ref{sec:controllability} is to provide enough material for a non-technical overview 
on the second subject in order to understand the results on
the first (where the difficulties come from). Mathematical
details are postponed to Sections \ref{sec:lie-algebra-block}
and \ref{sec:strong-contr}, while results on cavity systems are presented in overview in
Section~ \ref{sec:atoms-cavity}.   

\section{Controllability}
\label{sec:controllability}

The control of quantum systems poses considerable mathematical challenges when
applied to infinite dimensions. Basically, they arise from the fact that
anti-selfadjoint operators (recall that according to Stone's Theorem~\cite[VIII.4]{ReedSimonI}, 
they are generators of strongly continuous, unitary one-parameter groups) 
do neither form a Lie algebra nor even a vector space. Or seen on the group level, 
the group of unitaries equipped with the strong operator topology is a topological group yet not a Lie group.
So whenever strong topology has to be invoked, controllability cannot be assessed via a system Lie algebra. 
Thus we address these challenges on the group level by employing the controlled time evolution of the quantum system
in order to approximate unitary operators,  the action of which is measured with respect
to arbitrary, but finite sets of vectors. This is formalized in the notion of \emph{strong controllability}
(see Section~\ref{sec:strong-contr-1}) introduced here as a generalisation of  pure-state controllability 
already discussed in the literature. 
Central to our discussion are abelian symmetries. Assuming that all but one of our Hamiltonians 
observe such an ablian symmetry, we systematically analyze the infinite-dimensional control system
in its block-diagonalized basis. We obtain strong controllability (beyond pure-state controllability) 
if one of the Hamiltonian breaks this abelian symmetry and some 
further technical  conditions are fullfilled.

\subsection{Time evolution}
\label{sec:time-evolution}

We treat control problems of the form 
\begin{equation} \label{eq:1}
   \dot{\psi}(t) = \sum_k u_k(t) H_k \psi(t) = H(t) \psi(t)
\end{equation}
where the $H_k$ with $k\in\{1,\dots,d\}$ are selfadjoint control Hamiltonians on an
infinite-dimensional, separable Hilbert space $\mathcal{H}$ and the controls $u_k: \R \to \R$  are piecewise-constant control
functions. 
Since $\mathcal{H}$ is infinite-dimensional, the operators $H_k$ are usually
only defined on a dense subspace $D(H_k) \subset \mathcal{H}$ called the \emph{domain}
of $H_k$, the only exceptions being those $H_k$ which are bounded. However,
in this context,
control problems where all $H_k$ are bounded are not very interesting from a
physical point of view. In other words, there is no way around considering
those domains and many difficulties of control theory in infinite dimensions
arises from this fact\footnote{Note that domains of unbounded operators are
  not just a mathematical pedantism. The domain is a crucial part of the
  definition of an operator and contains physically relevant information. A
  typical example is the Laplacian in a box which requires \emph{boundary
    conditions} for a complete description. Up to a certain degree, domains can
  be regarded as an abstract form of boundary conditions (possibly at
  infinity).} .

We will also assume that Eq.~\eqref{eq:1} will have unique solutions for all
initial states $\psi_0 \in \mathcal{H}$ and all times $t$. So for
each pair of times $t_1 < t_2$ there is a unitary propagator 
  $U(t_1,t_2)\psi_0 = \mathcal{T} \int_{t_1}^{t_2} \exp(-i t H(t)) \psi_0$,
where $\mathcal{T}$ denotes time ordering. Observe that this condition is usually
\emph{not} satisfied, not even if the $H_k$ share a joint domain of essential
selfadjointness. Fortunately, the systems we are going to study do not
show such pathological behavior. Yet, a minimalistic way to avoid this problem
would be to restrict to control functions where only one $u_k$ is different from 0
at each time $t$. In this case the propagator $U(t_1,t_2)$ is just a
concatenation of unitaries $\exp(i t H_k)$ which are guaranteed to exist due
to selfadjointness of the $H_k$.

\subsection{Pure-state controllability}

A key-issue
in quantum control theory is \emph{reachability}: Given two pure states
$\psi_0$, $\psi \in \mathcal{H}$, we are looking for a time $T>0$ and control functions
$u_k$ such that $\psi = U(0,T)\psi_0$. In infinite dimensions, however, this
condition is too strong, since there might be states which can be reached only
in infinite time, or not at all. Yet, one may 
find a reachable state
``close by'' with arbitrary small control error. Therefore we will  call
$\psi$ reachable from $\psi_0$ if for all $\epsilon>0$ there is a
\emph{finite} time $T>0$
and control functions $u_k$ such that 
$  \| \psi - U(0,T)\psi_0\| < \epsilon$
holds. Accordingly, we will call the system (\ref{eq:1}) \emph{pure-state
  controllable}, if each pure state~$\psi$ can be reached from one~$\psi_0$
(and, by unitarity, also vice versa). 

Since pure states are described by one-dimensional projections, two state
vectors describe the same state if they differ only by a global phase. Hence the
definition just given is actually a bit too strong. There are several ways
around this problem, like using the trace norm distance of $\kb{\psi}$ and
$\kb{\psi_0}$ rather then the norm distance of $\psi$ and $\psi_0$. For our
purposes, however, the most appropriate method is to assume that the unit
operator $\unity$ on $\mathcal{H}$ is always among the control Hamiltonians. This 
may appear somewhat arbitrary, but it helps to avoid problems with determinants and traces on
infinite-dimensional Hilbert spaces, which  otherwise would arise.

\subsection{Strong controllability}
\label{sec:strong-contr-1}

Next, the analysis shall be lifted to the level of operators, i.e.\ to 
unitaries~$U$ from the group $\mathcal{U}(\mathcal{H})$ of unitary operators
on the Hilbert space $\mathcal{H}$ such that a time $T>0$ and control functions $u_k$ exist with 
$U= U(0,T)$. As in the last paragraph, this has to be generalized to an
approximative condition again. The best choice---mathematically as well as from a
practical point of view---is approximation in the \emph{strong sense}: We
look for unitaries $U$ such that for each set of (not necessarily
orthonormal or linearly independent) vectors $\psi_1, \dots, \psi_f \in
\mathcal{H}$ and each $\epsilon > 0$, there exists a time $T>0$ and control
functions $u_k$ such that
\begin{equation} \label{eq:26}
  \| [U - U(0,T)]\, \psi_k\| < \epsilon\, \text{ for all }\, k\in\{1,\dots,f\}.
\end{equation}
In other words, we are comparing $U$ and $U(0,T)$ only on a finite set of
states, and the worst-case error one can get here is bounded by $\epsilon$. We
will call the control system (\ref{eq:1}) \emph{strongly controllable} if each
unitary $U$ can be approximated that way.
(NB, in strong controllability, one again has the choice of one single
joint global phase factor.)

Clearly, strong controllability implies pure-state controllability. To see this, choose an
arbitrary but fixed $\psi_0 \in \mathcal{H}$. For each $\psi \in
\mathcal{H}$, there is a unitary $U$ with $U\psi_0 = \psi$. Hence strong
controllability implies $\|\psi - U(0,T)\, \psi_{0}\| = \| [U - U(0,T)]\, \psi_0\| < \epsilon$. 

\subsection{The dynamical group $\mathcal{G}$}
\label{sec:dynamical-group}

Strong controllability is concept-wise related to the
strong operator topology \cite[VI.1]{ReedSimonI} on the group
$\mathcal{U}(\mathcal{H})$ of unitary operators on $\mathcal{H}$.
To this end, consider the sets
\begin{equation} \label{eq:66}
  \mathcal{N}(U;\psi_1,\dots,\psi_f;\epsilon) = \{ V \in
  \mathcal{U}(\mathcal{H})\,|\,\|(V-U)\, \psi_k\| < \epsilon\; \text{ for all } \; k\in\{1,\dots,f\}
  \}.  
\end{equation}
They form a neighborhood base for the strong topology, and we will call them
\emph{(strong) $\epsilon$-neighborhoods}. The condition in Eq.~\eqref{eq:26} can now 
be restated as: Any $\epsilon$-neighborhood of $U$ contains a time-evolution
operator $U(0,T)$ for appropriate time $T$ and control functions $u_k$. In turn, this
can be reformulated as: $U$ is an \emph{accumulation point} of the set
$\tilde{\mathcal{G}}$ of all unitaries $U(0,T)$. The set of all accumulation
points of $\tilde{\mathcal{G}}$ (which contains $\tilde{\mathcal{G}}$ itself)
is a strongly closed subgroup\footnote{There is a subtle point here: The
  group $\mathcal{U}(\mathcal{H})$ is not
  strongly closed as a subset of the bounded operators $\mathcal{B}(\mathcal{H})$. Actually its
  strong closure is the set of all isometries;
  cf.\ \cite[Prob. 225]{halmos82}. Hence whenever we talk about strongly closed groups
  of unitaries, this has to be understood as the closure in the
  restriction of the strong topology to $\mathcal{U}(\mathcal{H})$ (which
  coincides with the restriction of the weak topology).}
of $\mathcal{U}(\mathcal{H})$, which we will
call the \emph{dynamical group} $\mathcal{G}$ generated by control Hamiltonians
$H_k$ with $k\in\{1,\dots,d\}$. If we choose the
controls  as described in Subsection~\ref{sec:time-evolution} (i.e.\ piecewise
constant and only one $u_k$ different from zero at each time), $\mathcal{G}$
is just the smallest strongly closed subgroup of $\mathcal{U}(\mathcal{H})$
that contains all $\exp(itH_k)$ for all $k\in\{1,\dots,d\}$ and all $t \in
\Bbb{R}$. Note that it contains in particular all unitaries that can
be written as a strong limit s-$\lim_{T \rightarrow \infty} U(0,T) $. In
finite dimensions, $\mathcal{G}$ can be calculated via its system algebra,
i.e.\ the Lie algebra $\mathfrak{l}$ generated by the $i H_k$, since each $U
\in \mathcal{G}$ can be written as $U=\exp(H)$ for an $H \in \mathfrak{l}$. 

In infinite dimensions, however, several difficulties can occur. First, unbounded
operators $H_k$ are only defined on a dense domain $D(H_k) \subset
\mathcal{H}$. The sum $H_k + H_j$ is therefore only defined on the 
intersection $D(H_k) \cap D(H_j)$ and the commutator even only on a subspace
thereof. There is no guarantee that $D(H_k) \cap D(H_j)$ contains more than
just the zero vector. In this case, the Lie algebra cannot even be defined.  

The minimal requirement to get around this difficulty is the existence of a
joint dense domain $D$, i.e.\ $D \subset D(H_j)$ and $H_j D \subset D$ for all
$j$. However, even then we do not know whether $\mathcal{G}$ can be generated from
$\mathfrak{l}$ in terms of exponentials. In general, it is impossible to
define some $\exp(H)$ for all $H \in \mathfrak{l}$.

There are several ways to deal with these problems. One is to consider cases
where the $H_k$ generate (i) a finite-dimensional Lie algebra and admit (ii) a
common, invariant, dense domain consisting of analytic vectors \cite{HTC83,LTCC05}. In
this case the exponential function is defined on all of $\mathfrak{l}$, and we
can proceed in analogy to the finite-dimensional case. The problem is that the
group $\mathcal{G}$ will become a finite-dimensional Lie group and its orbits through a
vector $\psi \in \mathcal{H}$ are finite-dimensional as well. Hence, we never
can achieve full controllability. This approach is well studied; cf.\ \cite{HTC83,LTCC05}
and  references therein.

Another possibility which includes the possibility to study an
infinite-dimensional Lie algebra $\mathfrak{l}$ is to restrict to bounded generators $H_k$. In this
case, one can define $\mathfrak{l}$ as a norm-closed subalgebra of the Lie algebra
$\mathcal{B}(\mathcal{H})$ of bounded operators, and one ends up with a Banach-space theory which
works almost in the same way as the finite-dimensional analog; cf.~\cite{Lang96}
for details. Although this is a perfectly reasonable approach from the mathematical
point of view, it is not very useful for physical applications, since in most
cases at least some of the $H_k$ are unbounded.

In this paper, we will  thus consider a different approach which splits the
generators into two classes. The first $d{-}1$ generators $H_1, \ldots, H_{d-1}$ admit an abelian
symmetry and can be treated---with Lie-algebra methods---along the lines
outlined in the next subsection. Secondly, the last generator $H_d$ breaks this symmetry and achieves full
controllability with a comparably simple argument.
The details will be explained
in Section~\ref{sec:lie-algebra-block} and \ref{sec:strong-contr}. 

\subsection{Abelian symmetries}
\label{sec:abelian-symmetries}

One way to avoid the problem described in the last subsection, arises if the
control system admits symmetries. In this section, we will only sketch the
structure, while the details are postponed to Sect. \ref{sec:lie-algebra-block}. 

Let us consider the case of a 
$\mathrm{U}(1)$-symmetry\footnote{The generalization to multiple charges,
  i.e.\ a $\mathrm{U}(1)^N$, is straightforward.}, i.e.\ a 
  (strongly continuous) unitary representation $z \mapsto \pi(z) \in \mathcal{U}(\mathcal{H})$
  of the abelian group $\mathrm{U}(1)$ on
$\mathcal{H}$ 
where $\mathcal{U}(\mathcal{H})$ denotes the group of unitaries on
$\mathcal{H}$. It can be written in terms of a selfadjoint operator $X$ with
pure point spectrum consisting of (a subset of) $\Z$ as
  $\mathrm{U}(1) \ni z = e^{i\alpha} \mapsto \pi(z) = \exp(i \alpha X) \in
  \mathcal{U}(\mathcal{H})$. 
If we denote the eigenprojection of $X$ belonging to the eigenvalue $\mu\in \Z$ as
$X^{(\mu)}$ (allowing the case $X^{(\mu)}=0$ if $\mu$ is not an
eigenvalue of $X$) we get a block-diagonal decomposition of $\mathcal{H}$ 
in the symmetry-adapted basis as
\begin{equation} \label{eq:37}
  \mathcal{H} = \bigoplus_{\mu=-\infty}^\infty \mathcal{H}^{(\mu)} \, \text{ with }\,
  \mathcal{H}^{(\mu)}=X^{(\mu)} \mathcal{H},
\end{equation}
and we can rewrite $\pi(z)$ again as
  $\mathrm{U}(1) \ni z = e^{i\alpha} \mapsto \pi(z)  =
  \sum_{\mu=-\infty}^\infty e^{i \alpha \mu} X^{(\mu)} \in
  \mathcal{U}(\mathcal{H})$. 
Here we will make two assumptions representing substantial
restrictions of generality:
\begin{enumerate}
\item \label{item:4}
  All eigenvalues of $X$ are of finite multiplicity, i.e.\ the
  $\mathcal{H}^{(\mu)}$ are finite-dimensional. This is crucial for basically
  everything we will discuss in this paper.
\item 
  All eigenvalues of $X$ are non-negative. This assumption can be relaxed at
  certain points (e.g.\ all material in Sect. \ref{sec:commuting-operators} can
  be easily generalized). However, it helps to simplify the discussion at a
  technical level and all examples we are going to consider in the next
  section are of this form. 
\end{enumerate}
The first important consequence of 
(\ref{item:4}) concerns the space
of \emph{finite particle vectors}   
\begin{equation} \label{eq:28}
  \FPV = \{ \psi \in \mathcal{H} \, | \, X^{(\mu)}\psi = 0\ \text{for all but
    finitely many $\mu$}\}, 
\end{equation}
since it becomes (due to finite-dimensionality of $\mathcal{H}^{(\mu)}$)
a ``good'' domain for basically all unbounded operators appearing in this
paper. Moreover one gets the following theorem:
\begin{thm} \label{thm:6}
  Consider a strongly continuous representation $\pi$ of $\mathrm{U}(1)$ on
  $\mathcal{H}$ and the corresponding charge-type operator $X$. Then the following
  statements hold:
  \begin{enumerate}
  \item 
    A selfadjoint operator $H$ commuting with $X$ admits $\FPV$ as an
    invariant domain, i.e.\ $\FPV \subset D(H)$ and $H \FPV = \FPV$. Hence the
    space 
      $\mathfrak{u}(X) = \{ iH \, | \, H=H^*\, \text{ commuting with }\, X\}$
    is a Lie algebra with the commutator as its Lie bracket.
  \item 
    The exponential map is well defined on $\mathfrak{u}(X)$ and maps it onto
    the strongly closed subgroup
      $\mathcal{U}(X) = \{ U \in \mathcal{U}(\mathcal{H})\, | \, [U,\pi(z)]=0 \, \text{ for all }\, z
      \in \mathrm{U}(1)\}$
    of $\mathcal{U}(\mathcal{H})$.
  \item 
    The subalgebra $\mathfrak{l} \subset \mathfrak{u}(X)$ generated by a family of
    Hamiltonians $iH_1, \dots, iH_d \in \mathfrak{u}(X)$ is mapped by the
    exponential map into the dynamical group $\mathcal{G}$ of the corresponding control
    problem. The strong closure of $\exp(\mathfrak{l})$ coincides with $\mathcal{G}$.
  \end{enumerate}
\end{thm}

The basic idea behind this theorem, is that one can cut off the decomposition
(\ref{eq:37}) at a sufficiently high $\mu$ without sacrificing strong
approximations as described in Subsection~\ref{sec:strong-contr-1}. One only has
to take into account that the cut-off on $\mu$ has to become higher when the
approximation error decreases. This strategy allows for tracing a lot of 
calculations back to finite-dimensional Lie algebras. We will postpone a
detailed discussion of this topic---including the proof of Theorem~\ref{thm:6}---to 
Section~\ref{sec:lie-algebra-block}. 

The only additional material one needs at
this point, since it is of relevance for the next section, is a subgroup of $\mathcal{U}(X)$ and
its corresponding Lie algebra  which relates unitaries
with determinant one and their traceless generators. Since the $iH \in
\mathfrak{u}(X)$ are unbounded and not necessarily positive, it is difficult
to give a reasonable definition of tracelessness, and the determinant of $U \in
\mathcal{U}(X)$ runs into similar problems. However, the elements of
$U \in \mathcal{U}(X)$ and $iH \in \mathfrak{u}(X)$ are block diagonal with
respect to the decomposition of $\mathcal{H}$ given in (\ref{eq:37}). In other
words $U=\sum_{\mu} U^{(\mu)}$ and $H=\sum_{\mu} H^{(\mu)}$ are infinite sums of operators\footnote{Two small remarks
  are in order here: (i).~Infinite sums require a proper definition of
  convergence in an appropriate topology. In Section
  \ref{sec:lie-algebra-block}, this will be made precise. (ii).~Operator products of the form $X^{(\mu)} H
  X^{(\mu)}$ are potentially problematic if $H$ is unbounded and therefore
  only defined on a domain. In our case, however, $X^{(\mu)}$ projects onto
  $\mathcal{H}^{(\mu)}$, which is a subspace of the domain $\FPV$ on which $H$
  is defined.}, where
  $U^{(\mu)} = X^{(\mu)} U X^{(\mu)} \in \mathcal{U}(\mathcal{H}^{(\mu)})$,
  $H^{(\mu)} = X^{(\mu)} H X^{(\mu)} \in \mathcal{B}(\mathcal{H}^{(\mu)})$, and
$X^{(\mu)}$ denotes the projection onto the $X$-eigenspace
$\mathcal{H}^{(\mu)}$. Since all the $U^{(\mu)}$ and $H^{(\mu)}$ are operators
on finite-dimensional vector spaces, one can define
\begin{align}
  \mathcal{SU}(X) &:= \{ U \in \mathcal{U}(X)\, | \, \det U^{(\mu)} = 1\; \text{ for all } \;
  \mu \in \Bbb{Z}\}, \label{eq:42}
  \\ \mathfrak{su}(X) &:= \{ iH \in \mathfrak{u}(X)\, | \,
  \tr(H^{(\mu)}) = 0\; \text{ for all } \; \mu \in \Bbb{Z}\}.
\end{align}
Obviously, $\mathcal{SU}(X)$ is a (strongly closed) subgroup of
$\mathcal{U}(X)$ and $\mathfrak{su}(X)$ is a Lie subalgebra of
$\mathfrak{u}(X)$. The image of $\mathfrak{su}(X)$ under the exponential map
therefore coincides with $\mathcal{SU}(X)$. Note that $\mathcal{SU}(X)$ is
effectively an infinite direct product of groups $\mathrm{SU}(d^{(\mu)})$, if
$d^{(\mu)} = \dim \mathcal{H}^{(\mu)}$ and not the ``special'' subgroup of
$\mathcal{U}(X)$. 

\subsection{Breaking the symmetry}

To get a fully controllable system, one has to leave the group $\mathcal{U}(X)$,
which can be thought of as being represented as block diagonal, see Fig.~\ref{fig:1}~a.
To this end, we have to add control Hamiltonians that break the
symmetry. There are several ways of doing so, and a successful strategy depends
on the system in question (beyond 
the treatment of the symmetric part of
the dynamics captured in Theorem~\ref{thm:6}). Here, we will present a special
result which covers the examples discussed in the next section. The first step
is another direct sum decomposition of 
  $\mathcal{H} = \mathcal{H}_- \oplus \mathcal{H}_0 \oplus \mathcal{H}_+$, where 
  $\mathcal{H}_\alpha = E_\alpha \mathcal{H}$, with $\alpha\in\{+,0,-\}$  
are projections onto the subspaces
$\mathcal{H}_\alpha$ and should satisfy $[E_\alpha,X^{(\mu)}] = 0$. 
Let in the following $\N:=\{1,2,3,\ldots\}$ denote the set of positive
integers and define $\N_0:=\N\cup\{0\}$.
Hence for $\mu \in \N_0$ we can introduce the
projections $X_\pm^{(\mu)} = X^{(\mu)} E_\pm$ which we require to be non-zero. For the exceptional case $\mu=0$ 
 the relation $X_-^{(0)} = X^{(0)} E_- = X^{(0)}$ should hold. 
 Futhermore we write $X^{(\mu)}_0 = X^{(\mu)} E_0$ for the overlap of
$X^{(\mu)}$ and $E_0$ which can (in contrast to $X^{(\mu)}_\pm$) be equal to zero for
all $\mu$. The $X^{(\mu)}_\alpha$ are projections onto the subspaces
$\mathcal{H}^{(\mu)}_\alpha := X^{(\mu)}_\alpha \mathcal{H}$ satisfying
$X^{(\mu)} = X^{(\mu)}_- \oplus X^{(\mu)}_0 \oplus X^{(\mu)}_+$. 

\begin{defi} \label{def:1}
  A selfadjoint operator $H$ with domain $D(H)$ is called complementary to
  $X$, if there exists a decompositon  $\mathcal{H} = \mathcal{H}_- \oplus \mathcal{H}_0 
  \oplus \mathcal{H}_+$ 
  as defined above such that: 
  \begin{enumerate}
  \item
    $\mathcal{H}_0 \subset D(X)$ and $H\psi = 0$ for all $\psi \in \mathcal{H}_0$.
  \item \label{item:1}
    $\FPV \subset D(H)$ and for all $\mu > 1$ we have $H X_{+}^{(\mu+1)}\psi =
    X_-^{(\mu)} H \psi$. The corresponding operator $X_-^{(\mu)} H
    X_+^{(\mu+1)} \in \mathcal{B}(\mathcal{H})$ is a partial isometry with
    $X_+^{(\mu+1)}$ as its source and $X_-^{(\mu)}$ as its target projection.
  \item \label{item:2}
    Given the projection $F_{[0]} = X^{(0)}_{\phantom{+}} \oplus X^{(1)}_-$ and the
    corresponding subspace $\mathcal{H}_{[0]} = F_{[0]} \mathcal{H}$. The
    group generated by $\exp(i t H)$ with $t \in \Bbb{R}$ and those $U \in
    \mathcal{SU}(X)$ which commute with $F_{[0]}$ acts transitively on the
    space of one-dimensional projections in $\mathcal{H}_{[0]}$. 
  \end{enumerate}
\end{defi}

\begin{figure}[tb]
(a)\hspace{70mm} (b)\\[-5mm]
\centerline{
\begin{tabular}{c@{\hspace{20mm}}c}
\hspace{12mm}\includegraphics[width=.37\textwidth]{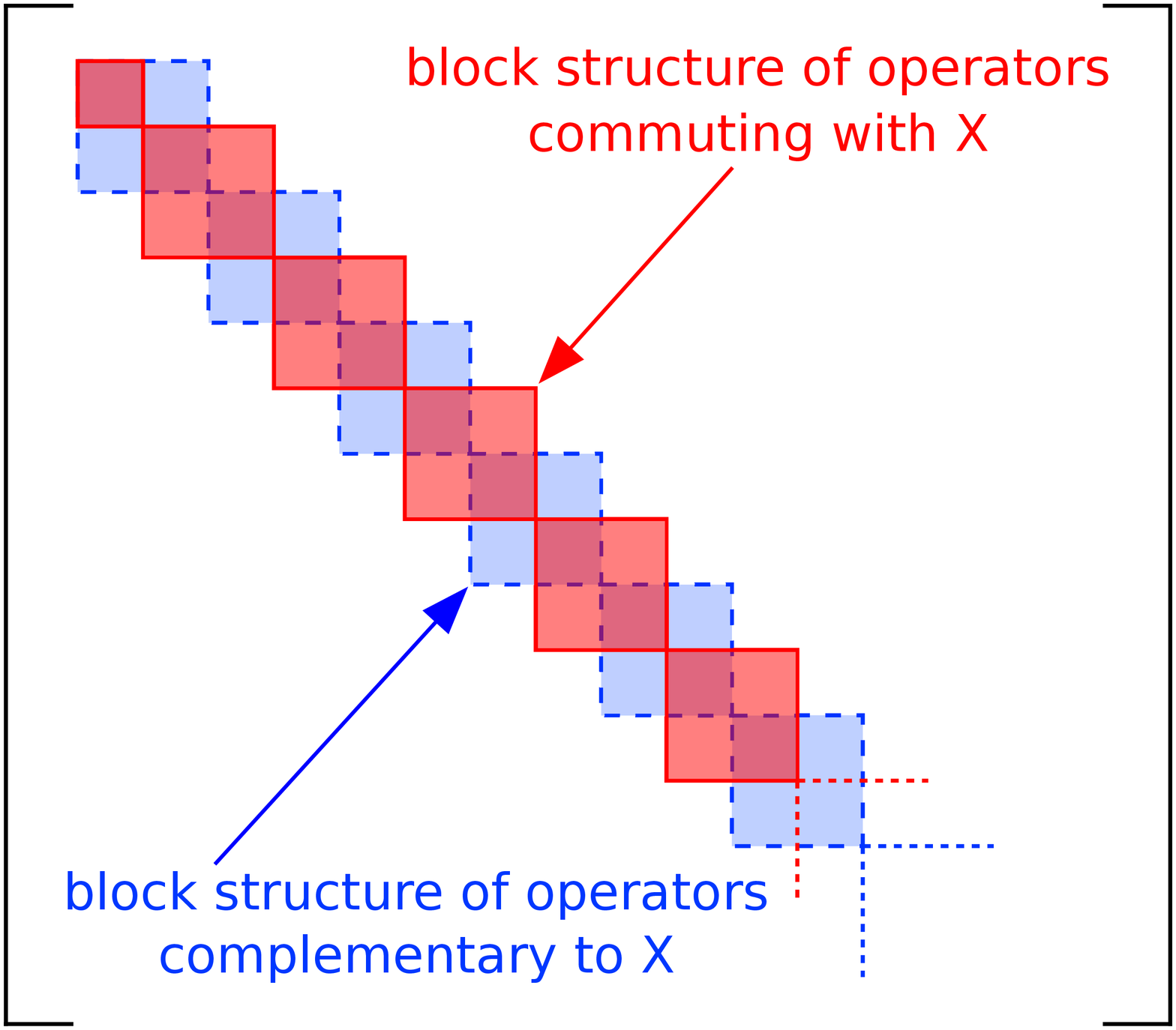} &
\raisebox{-2.3mm}{\includegraphics[width=.49\textwidth]{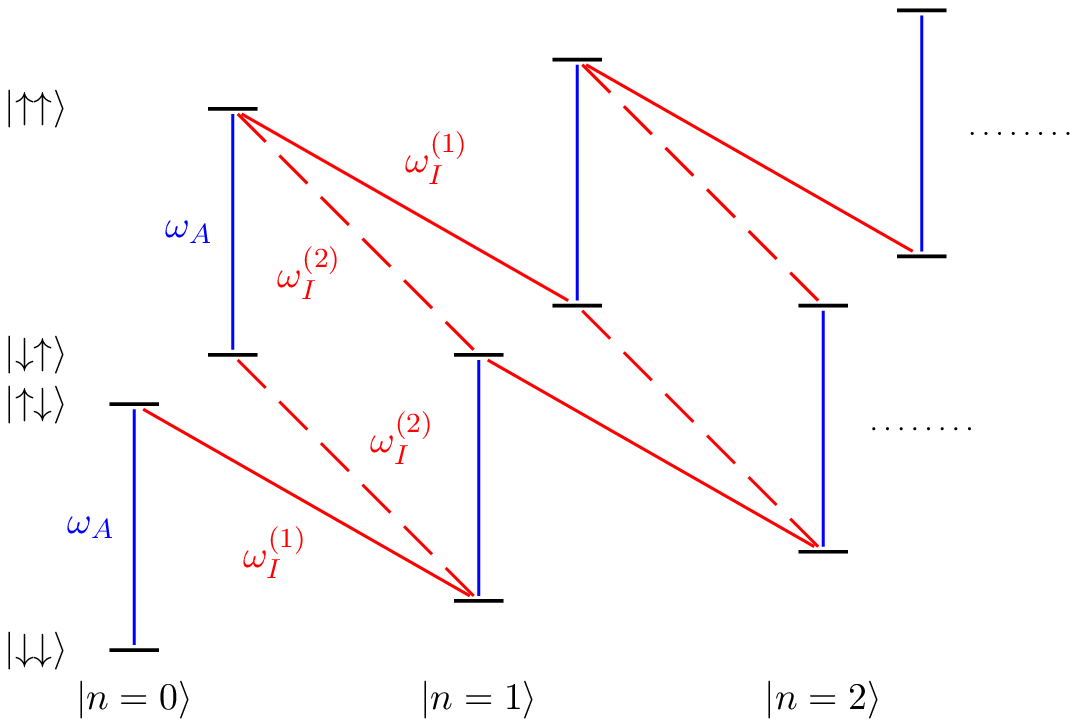} }
\end{tabular}}
\caption{\small \label{fig:1} (a) Block structure of operators in $\mathfrak{u}(X)$ (red) and of
    operators complementary to $X$ (blue) in the case where the projection
    $E_0$ vanishes. (b) Energy diagram  for the Jaynes-Cummings model (here two atoms in a cavity under individual controls
   $\omega_I^{(1)}$ and  $\omega_I^{(2)}$)
          with combined atom-cavity transitions matching the block structure of (a) given in red (see  Eqs.~(\ref{eq:15},\ref{eq:44})) since commuting with $X_1$ or $X_M$  of Eqs.~(\ref{eq:36}, \ref{eq:47}),
	and complementary transitions solely within the atoms given in blue  (see Eqs.~(\ref{eq:25},\ref{eq:45})). } 
\end{figure}

At first sight, the definition may look somewhat clumsy, but it allows for
proving a controllability result which covers all examples we are going to
present in the next 
section. We will state them here without a proof and postpone the latter to
Sect.~\ref{sec:strong-contr}.

\begin{thm} \label{thm:10}
  Consider a strongly continuous representation $\pi: \mathrm{U}(1) \to
  \mathcal{U}(\mathcal{H})$ with charge operator $X$ and a family of
  selfadjoint operators $H_1, \dots, H_d$ on
  $\mathcal{H}$. Assume that the following conditions hold:
  \begin{enumerate}
  \item 
    $H_1, \dots, H_{d-1}$ commute with $X$.
  \item 
    The dynamical group generated by $H_1, \dots, H_{d-1}$ contains
    $\mathcal{SU}(X)$.
  \item 
    The operator $H_d$ is complementary to $X$.
  \end{enumerate}
  Then the control system \eqref{eq:1} with Hamiltonians
  $H_0=\Bbb{1}, H_1, \dots, H_d$ is pure-state controllable.
\end{thm}
\begin{thm} \label{thm:9}
The
  control system (\ref{eq:1}) is even strongly controllable
  if in addition to the assumptions of Thm.~\ref{thm:10} the condition $\dim
  \mathcal{H}^{(\mu)} > 2$ holds for at least one $\mu \in \N_0$.
\end{thm}

\section{Atoms in a cavity}
\label{sec:atoms-cavity}

An important class of examples that can be treated along the lines described
in the last section are atoms interacting with the light field in a cavity. We
will discuss the case of $M$ two-level atoms interacting with one mode in
detail and consider three particular scenarios: one atom in Sect.~\ref{sec:one-atom}, 
individually controlled atoms in Sect.~\ref{sec:many-atoms-indiv},
and atoms under collective control in Sect.~\ref{sec:many-atoms-coll}.

\subsection{One atom}
\label{sec:one-atom}

Let us start with the special case $M=1$, i.e.\ one atom and one mode as
discussed in a number of previous publications  
mostly on pure-state controllability \cite{LE96,BBR10,YL07}. Our results go beyond this, in
particular because we are considering strong controllability not just pure-state controllability. 
The Hilbert space of the system is given by 
\begin{equation}
  \mathcal{H} = \C^2 \otimes \mathrm{L}^2(\R)
\end{equation}
and the dynamics is described by the well known Jaynes-Cummings Hamiltonian \cite{JC63}:
\begin{gather}
  H_{\mathrm{JC}} :=  \omega_A H_{\mathrm{JC},1} + \omega_I H_{\mathrm{JC},2} +
  \omega_C  H_{\mathrm{JC},3} \; \text{ with } \label{eq:9}\\
  H_{\mathrm{JC},1} := (\sigma_3 \otimes \unity)/2,\, H_{\mathrm{JC,2}}
  := (\sigma_+ \otimes a + \sigma_- \otimes a^*)/{2}, \,
  H_{\mathrm{JC},3} := \unity \otimes N \label{eq:15},
\end{gather}
where $\sigma_\alpha$ with $\alpha\in\{1,2,3\}$ are the Pauli matrices 
($\sigma_\pm=\sigma_1 \pm i \sigma_2$), $a,a^*$ denote
the annihilation and creation operator, and $N = a^*a$ is the number operator. The
joint domain of all these Hamiltonians is the space  
\begin{equation} \label{eq:40}
  D = \SP\{\ket{\nu}\otimes\ket{n}\, | \, \nu\in\{0,1\} \, \text{ and }\, n \in \N_0 \}
\end{equation}
with $\nu \in \C^2$ as canonical basis and $\ket{n} \in
\mathrm{L}^2(\R)$ as number basis (Hermite functions).  

We will assume that the frequencies $\omega_A$, $\omega_I$  and
$\omega_C$  can be controlled independently (or at least two of
them) such that we get a control system with control Hamiltonians
$H_{\mathrm{JC},j}$ where $j\in\{1,2,3\}$ 
corresponding to the lower half (1 atom) of the energy diagram in Fig.~\ref{fig:1}~b,
where we adopt the widely used convention of  
forcing the atom (spin) state $\ket{\negthickspace\uparrow}$ to be of \/`higher\/' energy than $\ket{\negthickspace\downarrow}$
to compensate for negative Larmor frequencies, see, e.g., the note in \cite[p.~144]{Haroche06}. 
The task is to determine the dynamical group $\mathcal{G}$. To this end, we use the
strategy described in Subsection~\ref{sec:abelian-symmetries}, which follows in
this particular case closely the exact solution of the Jaynes-Cummings model
\cite{JC63}. The charge-type operator $X_1$  (determining the block structure)
then takes the form
\begin{equation} \label{eq:36}
  X_1 = \sigma_3 \otimes \unity + \unity \otimes N,
\end{equation}
again with $D$ from (\ref{eq:40}) as its domain, which  in this case turns out 
to be identical to the space $D_{X_1}$ of finite-particle vectors. The operator $X_1$ is
diagonalized by the basis $\ket{\nu}\otimes\ket{n}$. It is convenient to
relabel these vectors in order to get
\begin{equation} \label{eq:27}
  \ket{\mu,\nu} = \ket{\nu} \otimes \ket{\mu - \nu} \in \mathcal{H}\,\text{ with }\, \mu =
  n+\nu \geq 0.
\end{equation}
In this basis, we have $X_1 \ket{\mu,\nu} = \mu \ket{\mu,\nu}$ and the
subspaces $\mathcal{H}^{(\mu)}$ from  (\ref{eq:37}) become
\begin{equation} \label{eq:38}
  \mathcal{H}^{(\mu)} = \SP\{ \ket{\mu,0},\ \ket{\mu,1} \}
\end{equation}
for $\mu>0$ and $\mathcal{H}^{(0)} = \C \ket{0,0}$ for $\mu=0$. 
The
space $D_{X_1}  \subset \mathcal{H}$ of finite-particle vectors turns out to be
identical with the domain $D$ from (\ref{eq:40}).

It is easy to see that the operators $H_{\mathrm{JC},j}$ from
Eq.~\eqref{eq:15} commute with $X_1$, and therefore we get $i H_{\mathrm{JC},j} \in
\mathfrak{u}(X_1)$.  A more detailed analysis, as will be given in
Section~\ref{sec:lie-algebra-block}, shows that  $i H_{\mathrm{JC},1}$  and $i H_{\mathrm{JC},2}$ generate
$\mathfrak{su}(X_1)$, and therefore we get according to Theorem~\ref{thm:6}: 
\begin{thm} \label{thm:3}
  The dynamical group $\mathcal{G}$ generated by $H_{\mathrm{JC},j}$ with $j\in\{1,2\}$ from
  Eq.~(\ref{eq:15}) coincides with the group $\mathcal{SU}(X_1)$ defined in (\ref{eq:42}).
\end{thm}
To get a fully controllable system, apply Theorem~\ref{thm:9} to see that 
one has to add  a Hamiltonian which  breaks the symmetry.  A possible candidate is
\begin{equation} \label{eq:25}
  H_{\mathrm{JC},4} = \sigma_1 \otimes \unity \in \mathcal{B}(\mathcal{H}).
\end{equation}
If we define the spaces $\mathcal{H}_\alpha$ as
  $\mathcal{H}_- = \SP \{\ket{\mu,0}\, | \, \mu \in \N_0\}$, $\mathcal{H}_0 =
\{0\}$, and $\mathcal{H}_+ = \SP\{\ket{\mu,1}\, | \, \mu \in \N \}$ 
the operator $H_{\mathrm{JC},4}$ becomes complementary to $X_1$, which can be
easily seen since $\mathcal{H}^{(\mu)}_+ = \Bbb{C} \ket{\mu,1}$,
$\mathcal{H}^{(\mu)}_- = \Bbb{C} \ket{\mu,0}$, and $\mathcal{H}^{(\mu)}_0 =
\{0\}$. Hence, according to Thm. \ref{thm:10}, the control system with Hamiltonians 
of Eqs.~(\ref{eq:9},\ref{eq:15}).
\begin{equation} \label{eq:65}
  H_0 = \Bbb{1},\; H_1 = H_{\mathrm{JC},1},\; H_2 =
  H_{\mathrm{JC},2},\; H_3 = H_{\mathrm{JC},4}
\end{equation}
is pure-state controllable\footnote{We have
  omitted the Hamiltonian $H_{\mathrm{JC},3}$ since it is not needed for the
  result. However, it can be added as a drift term without changing the result.}, and we are
recovering a previous result from \cite{LE96,BBR10,YL07}. However, with our
methods we can go beyond this and prove even strong
controllability. Thm. \ref{thm:9} cannot be applied since $\dim
\mathcal{H}^{(\mu)} \leq 2$ for all $\mu$, but the analysis of
Sect. \ref{sec:strong-contr} will lead to an independent argument.
\begin{thm} \label{thm:1}
  The control problem (\ref{eq:1}) with Hamiltonians $H_j$ and $j\in\{0,\dots,3\}$ from
  Eq.~\eqref{eq:65} is strongly controllable.
\end{thm}
Hence any unitary $U$ on $\mathcal{H}$ can be approximated  by varying the
control amplitudes $u_1=\omega_A$ and $u_{2}=\omega_{I}$ in the
Hamiltonian $H_{\mathrm{JC}}$ of (\ref{eq:9}) plus flipping ground and excited
state of the atom in terms of $H_{\mathrm{JC},4}$ (with strength $u_3$)---both in an
appropriate time-dependent manner. The approximation has to be understood in
the strong sense as described in Eq.~(\ref{eq:26}). 

Finally, note that Theorem \ref{thm:1} implies that one can simulate (again in
the sense of strong approximations) any unitary $V \in
\mathcal{B}(\mathrm{L}^2(\R))$ operating on the cavity mode alone. One only
has to find controls $u_j$ such that $U(0,T)\, \phi \otimes \psi_k$ $\approx$
$\phi \otimes V \psi_k$ for a finite set of states $\psi_k$ of the cavity (and
an arbitrary auxiliary state $\phi$ of the atom). 

\subsection{Many atoms with individual control}
\label{sec:many-atoms-indiv}

Next, consider the case of many atoms interacting with the same mode,
and under the assumption that each atom (including the coupling with the cavity)
can be controlled individually. Such a scenario is relevant for experiments
with ion traps, if the number of ions is not too big as have been studied
since \cite{CC00,GRL+03,LBM+03}. 
The Hilbert space of the system is
\begin{equation}
  \mathcal{H} = (\C^2)^{\otimes M} \otimes \mathrm{L}^2(\R),
\end{equation}
where $M$ denotes the number of atoms. We define the basis
  $\ket{b}\otimes\ket{n} \in \mathcal{H}$ 
  where  
  $n \in \N_0$,
$  \ket{b} = \ket{b_1} \otimes \dots \otimes \ket{b_M}$,
 $ b = (b_1,\dots,b_M) \in \Z_2 \times \dots \times \Z_2 = \Z_2^M$,
and  the canonical basis $\ket{b_j} \in \C^2$ with $b_j\in\{0,1\}$. 
The control Hamiltonians become
\begin{equation}
\label{eq:44}
H_{\mathrm{IC},j} = \sigma_{3,j} \otimes \unity\,\text{ and }\,
H_{\mathrm{IC},M{+}j} = \sigma_{+,j} \otimes a + \sigma_{-,j} \otimes a^*  
\end{equation}
where 
$j\in\{1,\dots,M\}$ and
$  \sigma_{\alpha,j} = \unity^{\otimes (j-1)} \otimes \sigma_\alpha \otimes
  \unity^{\otimes (N-j)}$. As before,
$a$ and $a^*$ denote annihilation  and creation operator.
 The joint domain of all these operators is
\begin{equation} \label{eq:41}
  D = \SP \{ \ket{b}\otimes\ket{n}\, | \, b \in \Z_2^M\, \text{ and }\, n \in
  \N_0 \}, 
\end{equation}
with the basis $\ket{b}\otimes\ket{n}$ as defined above.
As depicted by the red parts  in Fig~\ref{fig:1},
all the $H_{\mathrm{IC},k}$ are invariant under the symmetry defined by the charge operator 
\begin{equation}\label{eq:47}
  X_M = S_3 \otimes \unity + \unity \otimes N\,\text{ with }\, S_3 = \sum_{j=1}^N \sigma_{3,j}
\end{equation}
where $N=a^*a$ denotes again the number operator and $D$ from (\ref{eq:41}) is the
domain of $X_M$. The eigenvalues of $X_M$ are $\mu \in
\N_0$ and the eigenbasis is given by 
\begin{equation} \label{eq:48}
  \ket{\mu,b} = \ket{b} \otimes \ket{\mu - |b|} \,\text{ for }\, |b|=\sum_{j=1}^M b_j \leq
  \mu. 
\end{equation}
In this basis, $X_M$ becomes
  $X_M \ket{\mu,b} = \mu \ket{\mu,b}$
and the eigenspaces $\mathcal{H}^{(\mu)}$ are
  $\mathcal{H}^{(\mu)} = \SP \{\ket{\mu,b} \,|\, b \in \Z_2^M$ with $ |b| \leq
  \mu\} $.
From now on, one may readily proceed as for one atom to arrive at the following analogy to
Theorem~\ref{thm:3}:
\begin{thm} \label{thm:8}
  The dynamical group $\mathcal{G}$ generated by $H_{\mathrm{IC},k}$ with $k\in\{1,\dots,2M\}$ from
  Eq.~\eqref{eq:44} coincides with group $\mathcal{SU}(X_M)$ of unitaries commuting with
  $X_M$.   
\end{thm}
To get strong controllability, one has to add again one Hamiltonian. As before
a $\sigma_1$-flip of one atom is sufficient (see the blue parts  in Fig~\ref{fig:1}), and
\begin{equation} \label{eq:45}
  H_{\mathrm{IC},2M+1} = \sigma_{1,1} \otimes \unity.
\end{equation}
 is complementary to $X_M$ with $\mathcal{H}_\alpha$
given by $\mathcal{H}_0 =\{0\}$,
$\mathcal{H}_- = \SP \{ \ket{\mu;0,b_2,\dots,b_M}\,|\, \mu \in \N_0,\,
|(b_2,\dots,b_M)| \leq \mu \}$, 
$\mathcal{H}_+ = \SP \{  \ket{\mu;1,b_2,\dots,b_M}\,|\, \mu \in \N,\,  
|(b_2,\dots,b_M)| < \mu \}$.
Obviously, all the conditions of Thm. \ref{thm:9} are satisfied such that one
gets
\begin{thm} \label{thm:7}
  The control problem (\ref{eq:1}) with $H_{\mathrm{IC},k}$ and  $k\in\{1,\dots,2M{+}1\}$
  from (\ref{eq:44}) and (\ref{eq:45}) is strongly controllable. 
\end{thm}
As a special case of this theorem, one can approximate
any unitary $U$ acting on the atoms alone, i.e.\ $U \in
\mathcal{U}((\C^2)^{\otimes M})$, by applying Theorem~\ref{thm:7} to $U
\otimes \unity$. That is, one can simulate~$U$ only by operations on one atom and
the interactions with the harmonic oscillator. This is used in
ion-trap experiments and is known as ``phonon bus''. 

\subsection{Many atoms under collective control}
\label{sec:many-atoms-coll}

Now one may modify the setup from the last section by considering again $M$ atoms
interacting with one mode, but assuming that one can control the atoms only
collectively rather than individually. In other words instead of the
Hamiltonians $H_{\mathrm{IC},j}$ and  $H_{\mathrm{IC},M{+}j}$ with $j\in\{1,\dots,M\}$ 
 from Eq.~\eqref{eq:44} one only has their sums
\begin{equation} \label{eq:31}
  H_{\mathrm{TC},1} = S_3 \otimes \unity\,\text{ and }\, H_\mathrm{{TC},2} = S_+ \otimes
  a + S_- \otimes a^*,
\end{equation}
where
$S_\alpha = \sum_{j=1}^M \sigma_{\alpha,j}$ and $\alpha\in\{1,2,3,\pm\}$,
combinded with the free evolution 
\begin{equation} \label{eq:33}
  H_{\mathrm{TC},3} = \unity \otimes N
\end{equation}
of the cavity.
As  before, all operators are defined on the domain $D$ from
(\ref{eq:41}). Note that one readily recovers the original setup from
Subsection~\ref{sec:one-atom} with Pauli operators $\sigma_\alpha$ replaced by
pseudo-spin operators $S_\alpha$. The multi-atom analogue of the Jaynes-Cummings
Hamiltonian, which can be formed from the $H_{\mathrm{TC},j}$ just defined, is called
Tavis-Cummings Hamiltonian \cite{TC68,TC69}.  

All the Hamiltonians in Eqs.~(\ref{eq:31}) and (\ref{eq:33}) are invariant
under the $\mathrm{U}(1)$-action generated by $X_M$ of
Eq.~\eqref{eq:47}. However, this is not the only symmetry, since all these
$H_{\mathrm{TC},j}$ are also invariant under the permutation of the
atoms. Therefore, one may no longer exhaust the group $\mathcal{SU}(X_M)$ as in
Theorem~\ref{thm:8} (since the following operators 
cannot be reached: those
commuting only with $X_M$ but not also with permutations
of the atoms). A minimal modification is to restrict the states of the
atoms to spaces on which permutation-invariant unitaries operate
transitively\footnote{An alternative strategy would be to treat 
  permutation symmetry in the same way as 
  $\mathrm{U}(1)$-symmetry. However, already the restriction to permutation-invariant
  states will turn out to be difficult enough.}. The most natural
choice is the symmetric tensor product
$
  (\C^2)^{\otimes M}_{\sym} \subset (\C^2)^{\otimes M},
$
i.e.\ the Bose subspace of $(\C^2)^{\otimes M}$. The preferred basis of
$(\C^2)^{\otimes M}_{\sym}$ is 
  $\ket{\nu} = \Sym_M \left( \ket{1}^{\otimes \nu} \otimes \ket{0}^{\otimes (M
      - \nu)} \right)$
with $\nu \in \{0, \dots, M\}$ and the projection 
$
  \Sym_M $ from $(\C^2)^{\otimes M}$ 
onto the symmetric subspace $(\C^2)^{\otimes M}_{\sym}$. In other words $\ket{\nu}$ is the unique, pure,
permutation-invariant state with $\nu$ atoms in the excited state $\ket{1}$
and $M {-} \nu$ ones in the ground state $\ket{0}$. Therefore,
$(\C^2)^{\otimes M}_{\sym}$ can be identified with the Hilbert space
$\C^{M+1}$ of a (pseudo-)spin-$M/2$ system. Its
basis $\ket{\nu}$, with $\nu\in\{0,\dots,M\}$ becomes the canonical basis. Combining this
with $\mathrm{L}^2(\R)$ for the cavity one gets
  $\mathcal{H}_{\sym} = \C^{M+1} \otimes \mathrm{L}^2(\R)$
as the new Hilbert space of the system. 

All the operators defined above ($H_{\mathrm{TC},1}, H_{\mathrm{TC},2},
H_{\mathrm{TC},3}$ and $X_M$) can be restricted to $\mathcal{H}_{\sym}$ (and in
slight abuse of notation we will re-use the symbols after restriction) and
their domain becomes  
\begin{equation} \label{eq:51}
  D_{\sym} = \SP \{ \ket{\nu} \otimes \ket{n}\, | \, \nu\in\{0,\dots,M\}\,\text{ and }\, n \in
  \N_0 \},
\end{equation}
which is just the projection of $D$ from (\ref{eq:41}), i.e.\ $D_{\sym} =
\Sym_M D$. The eigenbasis of $X_M$ now takes the form
  $\ket{\mu,\nu} = \ket{\nu}\otimes\ket{\mu-\nu}$ where  $\mu \in \N_0$  and $\nu
  < d_\mu = \min(\mu,M{+}1)$. 
For the $X_M$-eigenspaces, we get again $X_M \ket{\mu,\nu} = \mu \ket{\mu,\nu}$  and
\begin{equation} \label{eq:52}
  \mathcal{H}_{\sym}^{(\mu)} = \SP \{ \ket{\mu,\nu} \, | \, \nu\in\{0,\dots,d_\mu \}\}.
\end{equation}
Now one can proceed as the in the previous cases: The operators
$H_{\mathrm{TC},1}, H_{\mathrm{TC},2}, H_{\mathrm{TC},3}$ are (as operators on
$\mathcal{H}_{\sym}$) invariant under the action generated by $X_M$ and
therefore elements of $\mathfrak{u}(X_M)$. However, one still 
cannot exhaust all of $\mathcal{U}(X_M)$ (or $\mathcal{SU}(X_M)$). One only gets: 
\begin{thm} \label{thm:4}
  The dynamical group $\mathcal{G}$ generated by the operators $H_{\mathrm{TC},1},
  H_{\mathrm{TC},2}, H_{\mathrm{TC},3}$ from Eqs.~\eqref{eq:31} and
  (\ref{eq:33}) is a strongly closed subgroup of $\mathcal{U}(X_M)$. For each unitary $V
  \in \mathcal{U}(X_M)$ and each $\mu \in \N_0$ we can find an element $U \in \mathcal{G}$ such
  that    
    $U \psi^{(\mu)} = V \psi^{(\mu)}$
  holds for all $\psi^{(\mu)} \in \mathcal{H}_{\sym}^{(\mu)}$.
\end{thm}
In other words: As long as the charge $\mu$ is fixed, one can still approximate
any $V \in \mathcal{U}(X_M)$, but if one considers superpositions of different charges
this is no longer the case, i.e.\ there are $\psi \in D_{X_M}$ and $V \in \mathcal{U}(X_M)$
such $U \psi \neq V \psi$ for all $U \in \mathcal{G}$. We have checked the latter
explicitly with the computer algebra system Magma \cite{magma} for the case
$M=2$. 
To circumvent this problem, one has to add control Hamiltonians. 
Unfortunately, it seems that one has to add quite a lot. The best result we
have got so far is to replace the operators from Eqs.~(\ref{eq:31}) and
(\ref{eq:33}) by
\begin{gather} 
  H_{\mathrm{CC},k} = \bigl(\kb{k} - \kb{k{-}1}\bigr) \otimes \Bbb{1} \, \text{ with }\,
  k\in\{1,\dots,M\}, \nonumber \\
  H_{\mathrm{CC},M+1} = H_{\mathrm{TC},2} = S_+ \otimes a + S_- \otimes
  a^*\,\text{ and }\,  
  H_{\mathrm{CC},M+2}  = \bigl(\KB{0}{1} + \KB{1}{0}\bigr) \otimes \Bbb{1}. \label{eq:30}
\end{gather} 
The operators $H_{\mathrm{CC},k}$ with $k\in\{1,\dots,M+1\}$ commute with $X_M$ and
generate (as we will see in Sect.~\ref{sec:many-atoms-coll-1}) the Lie algebra
$\mathfrak{su}(X_M)$. In addition we have $H_{\mathrm{CC},M+2}$ which is
complementary to $X_M$ with Hilbert spaces 
$\mathcal{H}_+ = \SP \{ \ket{\mu;0}\, | \, \mu \in \N_0\}$,
$\mathcal{H}_- = \SP \{ \ket{\mu;1}\, | \, \mu \in \N \}$, and 
$\mathcal{H}_0 = \SP \{ \ket{\mu,\nu},\ | \, \mu \in \N,\, \mu > 2,\,
\nu\in\{3,\dots,\min(M,\mu)\} \}$.
 Note that we get an example for Def. \ref{def:1} with a non-trivial
 $\mathcal{H}_0$. Now one can apply Thms. \ref{thm:6} and \ref{thm:9} to get
 the analogues of Theorems \ref{thm:3} and \ref{thm:1}: 
\begin{thm} \label{thm:5} 
  The dynamical group $\mathcal{G}$ generated by $H_{\mathrm{CC},k}$ with $k\in\{1,\dots,M{+}1\}$ from
  Eq.~(\ref{eq:30}) coincides with the group $\mathcal{SU}(X_M)$ of unitaries commuting with
  $X_M$.   
\end{thm}
\begin{thm} \label{thm:2}
  The control problem (\ref{eq:1}) with $H_0=\Bbb{1}$  and 
  $H_{\mathrm{CC},k}$ for $k\in\{1,\dots, M{+}2\}$ from (\ref{eq:30}) 
  is strongly controllable.
\end{thm}
To be able to control all diagonal traceless operators $H_{\mathrm{CC},k}$,
with $k\in\{1,\dots,M\}$ is a very strong assumption. Unfortunately, a detailed analysis
including computer algebra indicates that we cannot recover Theorem
\ref{thm:2} with fewer resources. 

\section{A Lie algebra of block-diagonal operators}
\label{sec:lie-algebra-block}

The purpose of this section is to re-discuss abelian symmetries and to provide
technical details (in particular proofs) we omitted in 
Sections~\ref{sec:controllability} and \ref{sec:atoms-cavity}. To this end, 
re-use the notations already introduced in Section~\ref{sec:abelian-symmetries}. 
In particular, the abelian symmetry induces a block-diagonal decomposition
which,  in infinite dimensions, allows for defining a block-diagonal Lie algebra 
and its exponential map onto a block-diagonal Lie group;
see Propositions~\ref{prop:4.1} and \ref{prop:4.2}. We identify 
the set of all block-diagonal unitaries reachable by block-diagonal 
time evolutions in Proposition~\ref{prop:5}
as the strong closure of exponentials of block-diagonal Lie algebra elements.
A central result is Corollary~\ref{kor:1}, in which the question of controllability
for the block-diagonal system of infinite dimensions 
is reduced to analyzing controllability for all finite-dimensional blocks.
Using finite-dimensional commutator calculations one can now establish controllability
on the infinite-dimensional but block-diagonal space
for each of the three control systems analyzed.

\subsection{Commuting operators}
\label{sec:commuting-operators}

The first step is a closer look at the Lie algebra $\mathfrak{u}(X)$ and the
corresponding group $\mathcal{U}(X)$ introduced in Theorem~\ref{thm:6} (which we will prove
in this context). To this end, let us start with a unitary $U$ commuting with
the representatives $\pi(z)$, i.e.\ $[\pi(z),U]=0$ for all $z \in
\mathrm{U}(1)$. This is equivalent to 
  $U\psi = \sum_{\mu=0}^\infty U^{(\mu)}\psi^{(\mu)}$ for all $\psi \in
  \mathcal{H}$   with  $\psi^{(\mu)} := X^{(\mu)} \psi \in
  \mathcal{H}^{(\mu)}$
given a sequence of unitaries $U^{(\mu)}$  on the $\mu$-eigenspaces
$\mathcal{H}^{(\mu)}$ of $X$. Similarly one can consider a selfadjoint $H$
with domain $D(H)$ commuting with $X$. By definition\footnote{Note
  	that the identity $[X,Y]\,\psi = 0$ for all $\psi$ on a common dense domain is--in contrast
  	to popular belief--{\em not a proper definition} for two commuting selfadjoint
  	operators; cf.\ the discussion in \cite[VIII.6]{ReedSimonI}. Fortunately, such
  	pathological cases do not occur in our set-up.} 
this means the spectral projections
of $H$ commute with the $X^{(\mu)}$, which is equivalent to
\begin{equation}\label{eq:29}
  \FPV \subset D(H),\, H \FPV \subset \FPV \, \text{ and }\, H \psi =
  \sum_{\mu=0}^\infty H^{(\mu)} \psi^{(\mu)} \,\text{ for }\, \psi \in \FPV 
\end{equation}
with a sequence  of selfadjoint operators $H^{(\mu)}$ on the eigenspaces
$\mathcal{H}^{(\mu)}$ and the $\psi^{(\mu)}$ as defined above. 
The
$\mathcal{H}^{(\mu)}$ are finite-dimensional, and therefore the $H^{(\mu)}$ are
bounded. Hence the unboundedness of $H$ is inherited only from the unboundedness of
the sequence of norms $\|H^{(\mu)}\|$.  So it is easy to 
see that all elements of $\FPV$ are analytic for $H$ and therefore $\FPV$
becomes a domain of essential selfadjointness for $H$ (i.e.\ $H$ is uniquely
determined by its restriction to $\FPV$ as a consequence of Nelson's
analytic vector theorem \cite[Thm. X.39]{ReedSimonII}). Accordingly, we
will denote (in slight abuse of notation) the selfadjoint operator $H$ and
its restriction to $\FPV$ by the same symbol. This proves very handy when
introducing, on the set $\mathfrak{u}(X)$ of anti-selfadjoint operators commuting
with $X$, the structure of a Lie algebra by 
$(\lambda Q_1 + Q_2) \psi = \lambda Q_1 \psi + Q_2 \psi,\, [Q_1,Q_2] \psi
  = Q_1 Q_2 \psi - Q_2 Q_1 \psi$ for 
  $Q_1, Q_2 \in \mathfrak{u}(X)$, $\lambda \in \Bbb{R}$, and
  $\psi \in \FPV$.
The linear combination $\lambda Q_1 + Q_2$ and the commutator $[Q_1,Q_2]$ are
defined only on the joint domain $\FPV$ but since they are essentially
selfadjoint on it, their selfadjoint extensions exist and are uniquely
determined. This proves the first statement of Thm.~\ref{thm:6}, which we restate
as follows:

\begin{prop}\label{prop:4.1}
  A selfadjoint operator $H$ commuting with $X$ admits $\FPV$ as an invariant
  domain of essential selfadjointness. The space
    \begin{align}
      \mathfrak{u}(X) &= \big\{ iH \, | \, H=H^*\ \text{commuting with}\ X\big\}\\
      &= \big\{ iH \, | \, H \psi = \mbox{$\sum_\mu$} H^{(\mu)}
        \psi^{(\mu)},\ \psi \in \FPV,\ H^{(\mu)} = (H^{(\mu)})^* \in
        \mathcal{B}(\mathcal{H}^{(\mu)})\big\} 
    \end{align}
    becomes a Lie algebra with the commutator as its Lie bracket.
\end{prop}

Since all  $iH \in \mathfrak{u}(X)$ are anti-selfadjoint, they admit a
well-defined exponential map $\exp(i H)$. Boundedness of the $H^{(\mu)}$
together with Eq.~\eqref{eq:29} allows to express $\exp(iH)$ very
explicitly. More precisely one has 
\begin{equation} \label{eq:22}
  \exp(iH)\, \psi = \sum_{\mu=-\infty}^\infty \exp(i H^{(\mu)})\,
  \psi^{(\mu)}\, \text{ where }\, \psi^{(\mu)} = X^{(\mu)}\psi \in \mathcal{H}^{(\mu)}
\end{equation}
and $\exp(i H^{(\mu)}) = \sum_{n=0}^\infty (iH^{(\mu)})^n/(n!)$.
This shows that $\exp: \mathfrak{u}(X) \rightarrow \mathcal{U}(X)$ is well-defined and onto
as stated in Thm.~\ref{thm:6}, which we are now ready to prove: 
\begin{prop}\label{prop:4.2}
  The exponential map on $\mathfrak{u}(X)$ is well-defined and given in terms of
  Equation~\eqref{eq:22}. 
  It maps $\mathfrak{u}(X)$ onto
  the strongly closed subgroup
  \begin{align}
    \mathcal{U}(X) &=  \big\{ U \in \mathcal{U}(\mathcal{H})\, | \, [U,\pi(z)]=0 \, \text{ for all }\, z
    \in \mathrm{U}(1)\big\} \\
    &= \big\{ U \, | \, U  \psi = \mbox{$\sum_\mu$} U^{(\mu)} \psi^{(\mu)},\
      \psi \in \mathcal{H},\ 
      U^{(\mu)} \in \mathcal{U}(\mathcal{H}^{(\mu)})\big\} \,\text{ of }\, \mathcal{U}(\mathcal{H}).
  \end{align}
\end{prop}

\begin{proof}
  The only statement not yet proven is the closedness of $\mathcal{U}(X)$. To this end, we
  have to show that for any net $(U_\lambda)_{\lambda \in \mathcal{I}}$
  strongly converging to a bounded operator $U$ we have $U \in \mathcal{U}(X)$. 
  As $U_\lambda \in \mathcal{U}(X)$ we have $[\pi(z),U_\lambda] =0$ 
  for all $\lambda$. Due to strong continuity of the map $A \mapsto [\pi(z),A]$ and
  the convergence of the $U_\lambda$ to $U$  it follows that $[\pi(z),U]=0$.
  Hence $U$ decomposes into a strongly converging series $U = \sum_\mu
  U^{(\mu)}$ with $U^{(\mu)} \in \mathcal{B}(\mathcal{H}^{(\mu)})$, and for
  each fixed $\mu$ we get $\lim_\lambda U_\lambda^{(\mu)} = U^{(\mu)}$. Since
  $\mathcal{H}^{(\mu)}$  is finite-dimensional, the nets
  $(U_\lambda^{(\mu)})_{\lambda \in \mathcal{I}}$ converge in norm and therefore $U^{(\mu)} \in
  \mathcal{U}(\mathcal{H}^{(\mu)})$ which implies $U \in \mathcal{U}(X)$.
\end{proof}

Note that we actually proved more than what we stated. A strongly
convergent sequence (or net) of elements of $\mathcal{U}(X)$ cannot converge to an
isometry which is not unitary as well. Hence $\mathcal{U}(X)$ is strongly closed as a subset of
$\mathcal{B}(\mathcal{H})$---and not only as a subset of $\mathcal{U}(\mathcal{H})$
as generally is the case (cf.\ corresponding remarks in Sect.~\ref{sec:dynamical-group}).

The remaining statements in this subsection are devoted to the dynamical
group~$\mathcal{G}$ generated by a family of selfadjoint operators $H_1, \dots, H_d$.
Recall that we have introduced it as the smallest strongly closed subgroup of
$\mathcal{U}(\mathcal{H})$ containing all unitaries of the form $\exp(i t
H_k)$. If the $H_k$ are commuting with $X$, i.e.\ $iH_k \in \mathfrak{u}(X)$, then the
group $\mathcal{G}$ is a subgroup of $\mathcal{U}(X)$, and the simple structure of the latter
makes explicit calculations at least feasible. In the following, we show
how $\mathcal{U}(X)$ is related to the Lie algebra $\mathfrak{l}$ generated by the
$iH_k$. To this end, we need some additional notations. For each $K \in
\N$, $U \in \mathcal{U}(X)$, and $iH \in \mathfrak{u}(X)$, let us consider 
\begin{equation} \label{eq:62}
  U^{[K]} = \sum_{\mu=0}^K U^{(\mu)},\; H^{[K]} = \sum_{\mu=0}^K
  H^{(\mu)},\; \mathcal{H}^{[K]} = \bigoplus_{\mu=0}^K \mathcal{H}^{(\mu)}. 
\end{equation}
The operators $U^{[K]}$ and $H^{[K]}$ act on the finite-dimensional Hilbert space
$\mathcal{H}^{[K]}$. Therefore all operator topologies coincide and we can
apply the well-known finite-dimensional theory. The dynamical group $\mathcal{G}^{[K]}$
(generated by $H^{[K]}_k$ with $k\in\{1,\dots,d\}$) becomes a closed subgroup of the
unitary group $\mathcal{U}(\mathcal{H}^{[K]})$, which is a Lie group. Hence
$\mathcal{G}^{[K]}$ is a Lie group, too, and its Lie algebra $\mathfrak{l}^{[K]}$ is
generated by $iH^{[K]}_k$ with $k\in\{1,\dots,d\}$. Now, the crucial point is that one can
approximate the infinite-dimensional objects $\mathcal{G}$ and $\mathfrak{l}$
by the finite-dimensional $\mathcal{G}^{[K]}$ and $\mathfrak{l}^{[K]}$. To see this, the
first step is the following lemma.

\begin{lem} \label{lem:1}
  Consider the Lie algebras $\mathfrak{l} \subset \mathfrak{u}(X)$ and
  $\mathfrak{l}^{[K]} \subset \mathcal{B}(\mathcal{H}^{[K]})$ (with $K \in \N$) generated by
  $iH_1, \dots, iH_d$ and $iH^{[K]}_1, \dots, iH^{[K]}_d$,
  respectively. Each element $\tilde{Q} \in \mathfrak{l}^{[K]}$ can be written as
  $\tilde{Q} = Q^{[K]}$ for some element $Q \in \mathfrak{l}$. 
\end{lem}

\begin{proof}
  Since $\tilde{Q} \in \mathfrak{l}^{[K]}$, it 
  is equal to a linear
  combination $\sum_{\ell} c_{\ell} C_{\ell}(iH_{j_1}^{[K]},\dots,iH_{j_n}^{[K]})$ of repeated commutators 
  $C_{\ell}(iH_{j_1}^{[K]},\dots,iH_{j_n}^{[K]})$ 
  containing the elements
  $\{iH_{j_1}^{[K]},\dots,iH_{j_n}^{[K]}\}$ with $j_k \in\{1, \dots, d\}$. However,
  $\mathfrak{l}$ is generated by $iH_1, \dots, iH_k$ and 
  it contains the same
  commutators $C_{\ell}(iH_{j_1},\dots,iH_{j_n})$ yet with $H_j^{[K]}$ replaced by
  $H_j$. Hence one can form a linear combination $Q$ such that $Q^{[K]} =
  \tilde{Q}$ as stated.  
\end{proof}

Moreover, we now have the tools to prove the relation between the Lie algebra
$\mathfrak{l}$ and the dynamical group $\mathcal{G}$ already stated in Thm. \ref{thm:6}. 

\begin{prop} \label{prop:5}
   Consider again $iH_1, \dots, iH_d \in \mathfrak{u}(X)$ and the Lie algebra
   $\mathfrak{l}$ generated by them. Then the corresponding dynamical group $\mathcal{G}$
   coincides with the strong closure of $\exp(\mathfrak{l}) \subset \mathcal{U}(X)$. 
\end{prop}

\begin{proof}
  Each $U \in \mathcal{G}$ can be written as the limit of a net $(U_\lambda)_{\lambda
    \in \mathcal{I}}$ of operators $U_\lambda$, which are monomials in
  $\exp(it_kH_k)$ with $k\in\{1,\dots,d\}$ with appropriate times $t_k$. This implies in
  particular that the $U_\lambda$ commute with $\pi(z)$ for all $z$, and, by
  continuity, the same is true for $U$. Hence $U\in \mathcal{U}(X)$, and for each $K \in
  \N$ we can define $U^{[K]}$ which is the limit of the net
  $(U^{[K]}_\lambda)_{\lambda \in \mathcal{I}}$. The latter converges in norm
  (since $\mathcal{H}^{[K]}$ is finite-dimensional), and therefore $U^{[K]}
  \in \mathcal{G}^{[K]}$. This implies $U^{[K]} = \exp(Q_K)$ with $Q_K \in
  \mathfrak{l}^{[K]}$ as $\mathcal{G}^{[K]}$ is a Lie group and $\mathfrak{l}^{[K]}$
  its Lie algebra. 

  For $U$ to be in the strong closure of $\exp(\mathfrak{l})$, each strong
  $\epsilon$-neighborhood of $U$, i.e.\ the sets
  $\mathcal{N}(U;\psi_1,\dots,\psi_f;\epsilon)$ introduced in
  Eq.~\eqref{eq:66}, should contain  an element of $\exp(\mathfrak{l})$ 
  for all $\psi_1, \dots, \psi_f$ and all $\epsilon > 0$. 
  However, the unitary group is contained in the unit
  ball of $\mathcal{B}(\mathcal{H})$, and thus it is sufficient to consider
  only those $\mathcal{N}(U;\psi_1,\dots;\psi_f,\epsilon)$ with vectors
  $\psi_1, \dots, \psi_f$ from a dense subspace of $\mathcal{H}$;
  cf.\ \cite[I.3.1.2]{blackadar06}. Hence, in turn, it is sufficient to consider only
  neighborhoods with $\psi_j \in \FPV$. But then there is a $K \in \N$
  such that $\psi_j \in \mathcal{H}^{[K]}$ for all $j\in\{1,\dots,f\}$. Now take
  the operator $Q_K$ from the last paragraph and $\tilde{Q}_K \in \mathfrak{l}$
  with $\tilde{Q}_K^{[K]} = Q_K$, which exists due to Lemma \ref{lem:1}. By
  construction we have 
  $\| [U - \exp(\tilde{Q}_K)]\psi_j\| = \| [U^{[K]}
  - \exp(\tilde{Q}_K)^{[K]}]\psi_j\| =  \| [(U^{[K]} -
  \exp(\tilde{Q}_K^{[K]})]\psi_j\| 
  = \| [(U^{[K]} - \exp(Q_K)]\psi_j\| = 0 $   
  since $U^{[K]} = \exp(Q_K)$, as was also seen in the previous paragraph. Hence
  $\exp(\tilde{Q}_K) \in \mathcal{N}(U;\psi_1,\dots,\psi_f;\epsilon)$ which
  shows that $U$ is in the strong closure of $\exp(\mathfrak{l})$. This shows that 
  the dynamical group $\mathcal{G}$ is contained in the strong closure of $\exp(\mathfrak{l})$.

  Conversely, consider $\exp(Q)$ for $Q \in \mathfrak{l}$. We have to show
  that $\exp(Q)$ is in the dynamical group $\mathcal{G}$. To this end we observe, for each $K
  \in \N$, that $\exp(Q^{[K]}) = \exp(Q)^{[K]}$, which is obviously in
  $\mathcal{G}^{[K]}$. Hence there is a $U_K = \exp(iH_{j_1}^{[K]}) \cdots
  \exp(iH_{j_n}^{[K]})$ with $j_k \in \{1,\dots,d\}$ which is $\epsilon$-close (in
  norm) to $\exp(Q^{[K]})$. As in the last paragraph, this implies that
  $\tilde{U} = \exp(iH_{j_1})\cdots\exp(iH_{j_n})$ is in
  $\mathcal{N}(\exp(Q);\psi_1, \dots, \psi_f;\epsilon)$ provided $\psi_j \in
  \mathcal{H}^{[K]}$ for all $j\in\{1,\dots,f\}$. Hence $\exp(Q)$ is in the strong
  closure of the group of monomials in the $\exp(iH_j)$, but this is just the
  dynamical group $\mathcal{G}$. Since $\mathcal{G}$ is strongly closed, the strong closure of
  $\exp(\mathfrak{l})$ is contained in $\mathcal{G}$, too. Since we have shown the other
  inclusion before, the entire proposition is proven.
\end{proof}

Moreover, with this proposition the proof of Thm.~\ref{thm:6} is complete. -- The rest of
this subsection is devoted to analyzing a related question: If, in finite
dimension, two Lie algebras $\mathfrak{l}_1, \mathfrak{l}_2$ generate the same
group, then they are actually identical. However, in infinite dimensions this no longer true. 
Therefore, the next proposition is meant to decide if dynamical groups generated by two
different sets of Hamiltonians do in fact coincide.

\begin{prop} \label{prop:6}
  Consider two Lie algebras $\mathfrak{l}_1, \mathfrak{l}_2 \subset
  \mathfrak{u}(X)$. Assume that for each $Q \in \mathfrak{l}_1$ and each $K \in
  \N$, there is a $\tilde{Q} \in \mathfrak{l}_2$ such that $Q^{[K]} =
  \tilde{Q}^{[K]}$ holds (note that we can have different
    $\tilde{Q}$ for the same $Q$ but different $K$). Then
  $\exp(\mathfrak{l}_1)$ is contained in the strong closure of
  $\exp(\mathfrak{l}_2)$.  
\end{prop}

\begin{proof}
  One may readily use the same strategy as in the proof of Prop. \ref{prop:5}: If the given
  condition holds, one can find in each neighborhood
  $\mathcal{N}(\exp(Q);\psi_1,\dots,\psi_f;\epsilon)$ of $\exp(Q)$ with
  $\psi_1,\dots,\psi_f \in \FPV$ an $\exp(\tilde{Q})$ with $\tilde{Q} \in
  \mathfrak{l}_2$. Hence $\exp(Q)$ is in the strong closure of
  $\exp(\mathfrak{l}_2)$.  
\end{proof}

Inserting $\mathfrak{su}(X)$ for $\mathfrak{l}_2$ provides a useful criterion to
check whether the dynamical group $\mathcal{G}$ generated by $H_1, \dots, H_d
\in \mathfrak{su}(X)$ is as large as possible in the sense that $\mathcal{G} =
\mathcal{SU}(X)$. To this end, let us introduce the truncated versions
\begin{align}
  \mathfrak{su}^{[K]}(X) &= \{ Q^{[K]}\,|\,Q \in \mathfrak{su}(X)\} = \oplus_{\mu=0}^K\,
  \mathfrak{su}(\mathcal{H}^{(\mu)}),\nonumber \\ 
  \mathcal{SU}^{[K]}(X) &= \{U^{[K]}\,|\, U \in \mathcal{SU}(X)\}=
  \oplus_{\mu=0}^K\, \mathcal{SU}(\mathcal{H}^{(\mu)}),
\end{align}
where we have used for any finite-dimensional subspace $\mathcal{K}$ of
$\mathcal{H}$ the notations $\mathfrak{su}(\mathcal{K})$ for the Lie algebra of
traceless operators on $\mathcal{K}$ and similarly $\mathcal{SU}(\mathcal{K})$
for the Lie group of unitaries on $\mathcal{K}$ with determinant $1$. Note
that  elements of $\mathfrak{su}(\mathcal{K})$
and $\mathcal{SU}(\mathcal{K})$ have---as operators on $\mathcal{H}$---a finite rank and their support and range
are both contained in $\mathcal{K}$. 

\begin{kor} \label{kor:1}
  Consider Hamiltonians $iH_1, \dots, iH_d \in \mathfrak{su}(X)$, the corresponding
  dynamical group $\mathcal{G}$ and the generated Lie algebra
  $\mathfrak{l}$. If $\mathfrak{su}^{[K]}(X) = \mathfrak{l}^{[K]}$ holds for all
  $K \in \N$, then one finds $\mathcal{G} = \mathcal{SU}(X)$. 
\end{kor}

\begin{proof}
  Simple application of Props.~\ref{prop:5} and \ref{prop:6}.
\end{proof}

\subsection{One atom}
\label{sec:one-atom-1}

The material just introduced readily applies to the systems studied in
Sect.~\ref{sec:atoms-cavity}. This includes in particular the proofs of
Thms.~\ref{thm:3}, \ref{thm:8}, \ref{thm:4} and \ref{thm:5}. The first step is
again one atom interacting with a cavity (Sect.~\ref{sec:one-atom}). Hence the
Hilbert space is $\mathcal{H} = \Bbb{C}^2 \otimes \mathrm{L}^2(\Bbb{R})$ and
the $\mathrm{U}(1)$-symmetry under consideration is generated by the operator
$X_1 = \sigma_3 \otimes \Bbb{1} + \Bbb{1} \otimes N$ already defined in
(\ref{eq:36}). The domain of $X_1$ is $D$ from Eq.~\eqref{eq:40}, which is
identical to $D_{X_1}$ introduced in (\ref{eq:28}). 

The next step is to characterize the Lie algebra $\mathfrak{l}$ generated by
the control Hamiltonians $H_{\mathrm{JC},1}$ and $H_{\mathrm{JC},2}$ as defined in
(\ref{eq:15}). They admit $D=D_{X_1}$ as a joint common domain, and it is easy to
see that they commute with $X_1$ (in the sense introduced in the previous
subsection). Hence $\mathfrak{l} \subset \mathfrak{u}(X_1)$, and all the machinery
from Subsection~\ref{sec:commuting-operators} applies. This includes in particular
the block-diagonal decomposition of operators $A \in \mathfrak{u}(X_1)$ given in
Eq.~\eqref{eq:29}. In our case the subspaces $\mathcal{H}^{(\mu)}$ with $\mu \in
\N$ are given by (cf.\ Eq.~\eqref{eq:38}) $\mathcal{H}^{(\mu)} = \SP\{
\ket{\mu,0},\ \ket{\mu,1} \}$ using the basis $\ket{\mu,\nu} \in \mathcal{H}$
introduced in (\ref{eq:27}). For $\mu=0$, we get the one-dimensional space
$\mathcal{H}^{(0)} = \Bbb{C} \ket{0,0}$.
The restrictions
$H_{\mathrm{JC},j}^{(\mu)}$ of the operators
$H_{\mathrm{JC},j}$ to the subspaces $\mathcal{H}^{(\mu)}$ are given
by (for $\mu \geq 1$):
\begin{equation} \label{eq:67}
  H_{\mathrm{JC},1}^{(\mu)} = -  \varsigma_3^{(\mu)}/2 ,\, H_{\mathrm{JC},2}^{(\mu)} = \sqrt{\mu}
  \varsigma_1^{(\mu)},\, H_{\mathrm{JC},3} = (\mu + 1/2) \varsigma_0^{(\mu)} -
   \varsigma_3^{(\mu)}/2,
\end{equation}
where we have introduced the operators $\varsigma_{\alpha}=\sum_{\mu} \varsigma_{\alpha}^{(\mu)}$ 
with $\alpha \in \{0,\ldots,3\}$
via their projections
$\varsigma_0^{(\mu)} = \Bbb{1}^{(\mu)} = X^{(\mu)} = \kb{\mu,0} + \kb{\mu,1}$,
$\varsigma_1^{(\mu)} = \KB{\mu,0}{\mu,1} + \KB{\mu,1}{\mu,0}$,
$\varsigma_2^{(\mu)} = i \bigl(\KB{\mu,1}{\mu,0} - \KB{\mu,0}{\mu,1}\bigr)$, and
$\varsigma_3^{(\mu)} = \kb{\mu,0} - \kb{\mu,1}$.
Hence, for each fixed $\mu$, the operator $\varsigma_\alpha^{(\mu)}$ is just the
corresponding Pauli operator on $\mathcal{H}^{(\mu)}$ given in the basis
$\ket{\mu,0},\ket{\mu,1}$. We have used the core symbol $\varsigma$ rather than
$\sigma$ in order to avoid confusion with the operators $\sigma_\alpha \otimes
\Bbb{1}$ acting only on the atom. In addition we introduce the operators
$A_{\alpha,k} \in \mathfrak{u}(X_1)$ with $\alpha\in\{0,\dots,3\}$ and $k \in \N_0$ by
\begin{equation} \label{eq:68}
  A_{\alpha,k} = \sqrt{X_1} X_1^k \varsigma_\alpha\,\text{ for }\, \alpha \in \{1,2\},\, A_{3,k} =
  X_1^k \varsigma_3,\, A_{0,k} = X_1^{k}.
\end{equation}
In terms of the $A_{\alpha,k}$, now the $H_{\mathrm{JC},j}$ can readily be  re-expressed as
\begin{equation} \label{eq:39}
  H_{\mathrm{JC},1} = - A_{3,0}/2,\, H_{\mathrm{JC},3} = A_{1,0}/2,\, H_{\mathrm{JC},2} = A_{0,1} +
 (A_{0,0} - A_{3,0})/2.  
\end{equation}
The next lemma shows that the Lie algebra $\mathfrak{l}$ generated by the
$H_{\mathrm{JC},j}$ is spanned as a vector space by a subset of the
$A_{\alpha,k}$. 

\begin{lem} \label{lem:2}
  The Lie algebra $\mathfrak{l}$ generated by  $iH_{\mathrm{JC},j}$ with $j\in\{1,2\}$
  is spanned as a vector space by the operators $i A_{\alpha,k}$ with
  $\alpha\in\{1,2,3\}$ and
  $k \in \N_0$.
\end{lem}

\begin{proof}
  Obviously the operators $iA_{\alpha,k}$ are in $\mathfrak{su}(X_1)$. Hence, they
  span a subspace $\tilde{\mathfrak{l}} \subset \mathfrak{su}(X_1)$. 
  To prove that $\tilde{\mathfrak{l}}$ is a Lie subalgebra of
  $\mathfrak{su}(X_1)$ one only has to check that $[A_{\alpha,k},A_{\beta,j}] \in
  \tilde{\mathfrak{l}}$ for all $\alpha,\beta \in\{1,2,3\}$ and $j,k \in
  \mathbb{N}_0$. This follows easily, because the $A_{\alpha,k}$ are just
  products of powers of $X_1$ and the $\varsigma_\alpha$. But the latter are
  representatives of the Pauli operators. Hence
  \begin{equation} \label{eq:18}
    [A_{1,k},A_{2,\ell}] = 2 i A_{3,k+\ell+1},\, [A_{3,k},A_{1,\ell}] = 2 i
    A_{2,k+\ell},\, [A_{2,k}, A_{3,\ell}] = 2 i A_{1,k+\ell}, 
  \end{equation}
  All operators vanish in the case of $\mu=0$. Hence $\tilde{\mathfrak{l}}$ is a
  Lie algebra and Eq.~\eqref{eq:39} proves that $\mathfrak{l} \subset
  \tilde{\mathfrak{l}}$. 

  For proving $\tilde{\mathfrak{l}} = \mathfrak{l}$, one has to express the
  $A_{\alpha,k}$ for $\alpha\in\{1,2,3\}$ and $k \in \N_0$ 
  in terms of repeated commutators of the  $H_{\mathrm{JC},2}$ and $H_{\mathrm{JC},3}$. 
  By the commutation relations in Equation \eqref{eq:18} it is obvious  that
  $\tilde{\mathfrak{l}}$ is generated (as a Lie algebra) by $A_{\alpha,0}$ with
  $\alpha \in\{1,2,3\}$. Therefore, the statement follows from Eq.~\eqref{eq:39},
  which in turn shows that $A_{1,0}$ and $A_{3,0}$ are just $H_{\mathrm{JC},3}$ and
  $H_{\mathrm{JC},1}$, while $A_{2,0}$ can be derived from the commutator
  $[H_{\mathrm{JC},1},H_{\mathrm{JC},3}]$. 
\end{proof}

With this Lemma and the material developed in the last subsection, one can proceed to
determine the structure of the dynamical group generated by
$H_{\mathrm{JC},1}$ and $H_{\mathrm{JC},2}$. This is the content of
Thm.~\ref{thm:3}, which is restated (and proven) here as a proposition.

\begin{prop} \label{prop:7}
  The dynamical group generated by $H_{\mathrm{JC},1}$ and $H_{\mathrm{JC},2}$
  is equal 
  to $\mathcal{SU}(X)$. 
\end{prop}

\begin{proof}
  According to Prop.~\ref{prop:5} the dynamical group $\mathcal{G}$ is the strong
  closure of $\exp(H)$ with $H \in \mathfrak{l}$, i.e.\ the Lie algebra generated
  by $H_{\mathrm{JC},1}$ and $H_{\mathrm{JC},2}$, while $\mathcal{SU}(X)$ is the
  strong closure of $\exp(\mathfrak{su}(X))$. Hence, by Cor.~\ref{kor:1}
  we   have to show that the truncated algebras $\mathfrak{l}^{[K]}$ and
  $\mathfrak{su}^{[K]}(X)$ are identical. The inclusion $\mathfrak{l}^{[K]}
  \subset \mathfrak{su}^{[K]}(X)$ is trivial, since all the blocks
  $H_{\mathrm{JC},j}^{(\mu)}$ with $j\in\{1,2\}$ are traceless. To show the other
  inclusion, first note that 
  $\mathfrak{l}^{[0]} = \mathfrak{su}(X)^{[0]} = \{0\}$. Hence it is sufficient to
  check that for each fixed $0< \mu_0 \leq K$ and each $iH \in
  \mathfrak{su}^{[K]}(X)$ with $H^{(\mu)} = 0$ for $\mu \neq \mu_0$ there is
  an $i A \in \mathfrak{l}$ such that $i A^{(\mu_0)} = i H^{(\mu_0)}$ and
  $A^{(\mu)} = 0$ for all $0 < \mu \leq K$ with $\mu \neq \mu_0$.  The rest
  follows by linearity. 

  For constructing such an $A$, recall from Lemma \ref{lem:2} that $\mathfrak{l}$
  is spanned (as a vector space) by the $A_{\alpha,k}$ with $\alpha\in\{1,2,3\}$ and $k \in
  \N_0$. Now consider a polynomial $f$ in one real variable satisfying
  $f(\mu) = 0$ for all $0 < \mu \leq K$ with $\mu \neq \mu_0$ and $f(\mu_0) =
  1$. The operators 
  $B_{\alpha,f} = f(X) \sqrt{X} \varsigma_\alpha$ with $\alpha \in \{1,2\}$ and
  $B_{3,f} = f(X) \varsigma_3$
  are linear combinations of the $A_{\alpha,k}$, and they satisfy the
  condition 
  $B_{\alpha,f}^{(\mu)} = 0$ for all  $0 < \mu \leq K$ such that $\mu\neq\mu_0$ and $B_{\alpha,f}^{(\mu_0)} =
    c_\alpha \varsigma_\alpha^{(\mu_0)}$ 
 for a constant $c_\alpha$ given by $c_{1}=c_{2}=\sqrt{\mu_0}$ and $c_3 = 1$. But
 all traceless operators $H^{(\mu_0)} \in \mathcal{B}(\mathcal{H}^{(\mu_0)})$
 can be written as a linear combinations of the $\varsigma_\alpha^{(\mu_0)}$,
 which concludes the proof.  
\end{proof}

Before proceeding to the next subsection, consider the free
Hamiltonian of the cavity $H_{\mathrm{JC},3}$. We have omitted it from the
discussion of the dynamical group, and the reason can be seen easily from
(\ref{eq:18}): $H_{\mathrm{JC},2}$ differs from $\mathrm{H}_{\mathrm{JC},1}$
only by $X_1 + \Bbb{1}/2$ which commutes with all elements of
$\mathfrak{su}(X)$. Hence adding $H_{\mathrm{JC},3}$ as a control Hamiltonian
would just add a one-dimensional center to the dynamical group
$\mathcal{G}=\mathcal{SU}(X)$. For the same reason, $H_{\mathrm{JC},3}$ could
be easily added as a drift term. Any effect it may have can be undone by evolving
the system with $H_{\mathrm{JC},1}$, and the remaining relative phase between
sectors of different charge $\mu$ does not affect the
discussion of strong controllability in Sect.~\ref{sec:strong-contr}. 
Finally, let us remark that---due to the same reasons just discussed---we
could exchange $H_{\mathrm{JC},1}$ and $H_{\mathrm{JC},3}$ almost without
changes to the results of this subsection.

\subsection{Many atoms with individual control}

First, recall some notations from Sect.~\ref{sec:many-atoms-indiv}. The
Hilbert space is $\mathcal{H}_M = (\Bbb{C}^2)^{\otimes M} \otimes
\mathrm{L}^2(\Bbb{R})$ using the distinguished basis $\ket{\mu;\vec{b}}$ with $\vec{b} \in
\Bbb{Z}_2^M$ from Eq.~\eqref{eq:48}. The charge operator is $X_M = S_3
\otimes \unity + \unity \otimes N$, cf.\ Eq.~\eqref{eq:47}, with domain $D_M$
from Eq.~\eqref{eq:41}. In addition, let us introduce the re-ordered tensor
product (where $\ket{\mu,b_1,\dots,b_M} \in \mathcal{H}_M$ and $b\in\Bbb{Z}_2$)
\begin{equation}
  \ket{\mu,\vec{b}} \hat{\otimes}_k \ket{b} =
  \ket{\mu+b;b_1,\dots,b_{k-1},b,b_k,\dots,b_M} \in \mathcal{H}_{M+1}.
\end{equation}
The key result of this section is split into the following three lemmas,
which eventually will lead to a proof of Thm.~\ref{thm:8}. 

\begin{lem} \label{lem:5}
  The complexification $\mathfrak{su}_\Bbb{C}(\mathcal{H}_M^{(\mu)})$ of the
  real Lie algebra $\mathfrak{su}(\mathcal{H}_M^{(\mu)})$ is generated by 
  elements $\KB{\mu;\vec{b}}{\mu;\vec{c}}$ with $\vec{b},\vec{c} \in
  \Bbb{Z}_2^M$ satisfying $\vec{b} \neq \vec{c}$.
\end{lem}

\begin{proof}
  $\mathfrak{su}_\Bbb{C}(\mathcal{H}_M^{(\mu)})$ is isomorphic to the Lie algebra
  $\mathfrak{sl}(\mathcal{H}_M^{(\mu)})$ of traceless operators on
  $\mathcal{H}_M^{(\mu)}$. The $\KB{\mu;\vec{b}}{\mu;\vec{c}}$ with $\vec{b}
  \neq \vec{c}$ span the vector space of all $A \in
  \mathcal{B}(\mathcal{H}_M^{(\mu)})$ satisfying $\langle \mu; \vec{b}\,|\, A
  \,|\, \mu; \vec{b}\rangle = 0$ for all $\vec{b} \in \Bbb{Z}_2^M$  i.e.\ all
  operators which are off-diagonal in the basis $\ket{\mu;\vec{b}}$. The
  smallest Lie algebra containing this space is
  $\mathfrak{sl}(\mathcal{H}_M^{(\mu)})$. 
\end{proof}

\begin{lem}\label{lem_prev}
  The Lie algebra $\mathfrak{su}_\Bbb{C}(\mathcal{H}_{M+1}^{(\mu)})$ is generated by the union of
  the subalgebras $\mathfrak{su}_\Bbb{C}(\mathcal{H}_M^{(\mu-b)} \hat{\otimes}_k
  \ket{b})$ with $b \in \Bbb{Z}_2$ and $k\in\{1,\dots,M\}$.
\end{lem}

\begin{proof}
  First of all, note that (by definition) $\ket{\mu{-}b;\vec{b}} \in
  \mathcal{H}_M^{(\mu-b)}$. Hence $\ket{\mu{-}b;\vec{b}} \hat{\otimes}_k \ket{b} \in
  \mathcal{H}_{M+1}^{(\mu)}$ which shows that all the Hilbert spaces
  $\mathcal{H}^{(\mu-b)} \hat{\otimes}_k \ket{b}$ are subspaces of
  $\mathcal{H}_M^{(\mu)}$. According to the previous lemma, we have to show
  that operators $A = \KB{\mu;\vec{b}}{\mu;\vec{c}}$ with $\vec{b}, \vec{c} \in
  \Bbb{Z}_2^{M+1}$ and $\vec{b} \neq \vec{c}$ can be written as commutators from
  operators in the $\mathfrak{su}_\Bbb{C}(\mathcal{H}_M^{(\mu+b)} \hat{\otimes}_k
  \ket{b})$. We have to distinguish two cases: 
  In the first case, there is at least one $k\in\{1,\dots,M\}$ with $b_k = c_k = b$. If this holds,
  $A$ can be written as 
 $
    \KB{\mu{-}b;b_1,\dots,b_{k-1},b_{k+1},\dots,b_{M+1}}{\mu{-}b;c_1,\dots,c_{k-1},c_{k+1},\dots,c_{M+1}} 
    \otimes \kb{b}
  \in\mathfrak{su}_\Bbb{C}(\mathcal{H}_M^{(\mu-b)} \hat{\otimes}_k \ket{b})$. The second
  case arises if $b_k \neq c_k$ for all $k$. Now consider the commutator
  of the operators $B =
  \KB{\mu;\vec{b}}{\mu;b_1,c_2,\dots,c_{M+1}}$ and $C =
  \KB{\mu;b_1,c_2,\dots,c_{M+1}}{\mu;\vec{c}}$ obviously $A = [B,C]$, $B
  \in \mathfrak{su}_\Bbb{C}(\mathcal{H}_M^{(\mu-b_1)} \hat{\otimes}_1 \ket{b_1})$, and $C
  \in \mathfrak{su}_\Bbb{C}(\mathcal{H}_M^{(\mu-c_k)} \hat{\otimes}_k \ket{c_k})$ for
  $k>1$. This concludes the proof.
\end{proof}

\begin{lem} \label{lem:6}
  The Lie algebra $\mathfrak{su}(\mathcal{H}_{M+1}^{(\mu)})$ is contained in
  the Lie algebra $\mathfrak{g}$ generated by
  $\mathfrak{su}(\mathcal{H}_M^{(\mu)}) \hat{\otimes}_k \Bbb{1}$ and
  $\mathfrak{su}(\mathcal{H}_M^{(\mu-1)}) \hat{\otimes}_k \Bbb{1}$. 
\end{lem}

\begin{proof}
  First of all note that it is sufficient to prove the statement for the
  corresponding complexified Lie algebras
  $\mathfrak{su}_\Bbb{C}(\mathcal{H}_{M+1}^{(\mu)}) =
  \mathfrak{su}(\mathcal{H}_{M+1}^{(\mu)}) \oplus i (\mathcal{H}_{M+1}^{(\mu)})$
  and $\mathfrak{g}_\Bbb{C} = \mathfrak{g} \oplus i \mathfrak{g}$, since we
  get the original statement back by restricting the inclusion
  $\mathfrak{su}_\Bbb{C}(\mathcal{H}_{M+1}^{(\mu)}) \subset
  \mathfrak{g}_\Bbb{C}$ to anti-selfadjoint elements on both sides.  

  The elements of $\mathfrak{su}_\Bbb{C}(\mathcal{H}^{(\mu)}) \hat{\otimes}_k \Bbb{1}$
  are of the form $A = a \hat{\otimes}_k \kb{0} + a \hat{\otimes}_k \kb{1}$ with
  $a \in \mathfrak{su}_\Bbb{C}(\mathcal{H}^{(\mu)})$. We will show that both summands
  are elements of $\mathfrak{g}_\Bbb{C}$, i.e.\ $a \hat{\otimes}_k \kb{b} \in
  \mathfrak{g}_\Bbb{C}$ for $b\in\{0,1\}$. The same holds for $\mu{-}1$. The statement then
  follows from Lemma~\ref{lem_prev}.

  Use again Lemma~\ref{lem:5} and choose $a = \KB{\mu;\vec{b}}{\mu;\vec{c}}$
  with  $\vec{b}, \vec{c} \in \Bbb{Z}_2^M$ and $\vec{b} \neq \vec{c}$. 
  We rewrite $A = a \hat{\otimes}_k \kb{0} + a \hat{\otimes}_k
  \kb{1}$ as 
  \begin{align}
  &\KB{\mu;(b_1,\dots,b_k,0,b_{k+1},\dots,b_M)}{\mu;(c_1,\dots,c_k,0,c_{k+1},\dots,c_M)}\nonumber \\ + &
  \KB{\mu+1;(b_1,\dots,b_k,1,b_{k+1},\dots,b_M)}{\mu+1;(c_1,\dots,c_k,1,c_{k+1},\dots,c_M)}.
  \end{align} 
  Moreover,  
    $\vec{b}_0:= (b_2,\dots,b_k,0,b_{k+1},\dots,b_M)$, $\vec{b}_1:=
    (b_2,\dots,b_k,1,b_{k+1},\dots,b_M)$, 
    $\vec{c}_0:= (c_2,\dots,c_k,0,c_{k+1},\dots,c_M)$, and $\vec{c}_1:=
    (c_2,\dots,c_k,1,c_{k+1},\dots,c_M)$
  allows us to simplify 
  \begin{equation} \label{eq:46}
     A=(\KB{\mu{-}b_1;\vec{b}_0}{\mu{-}c_1;\vec{c}_0}
     + \KB{\mu{-}b_1{+}1;\vec{b}_1}{\mu{-}c_1{+}1;\vec{c}_1}) \hat{\otimes}_1
     \KB{b_1}{c_1}.
  \end{equation}
  Next, consider a second operator $B = (\kb{\mu{-}c_1;\vec{c}_0}
  - \kb{\mu{-}c_1;\vec{c}_1}) \hat{\otimes}_1 \Bbb{1}$ and assume that $M > 1$
  holds. Then there is a $\ell\in\{1,\dots,M\}$ with $b_\ell \neq c_\ell$. Without loss of
  generality one can assume that $\ell\neq 1$ (otherwise rewrite $A$ in
  (\ref{eq:46}) as $\tilde{A} \hat{\otimes}_j \KB{b_j}{c_j}$ with another
  index $j$). The commutator now equals 
     $[A,B] = \KB{\mu-b_1;\vec{b}_0}{\mu-c_1;\vec{c}_0} \hat{\otimes}_1
     \KB{b_1}{c_1} = a \hat{\otimes}_k \kb{0}$. 
  If $M=1$ one has two possible cases: either $b=0$ and $c=1$ or $b=1$ and $c=1$. In the first
  case choose $B=(\kb{\mu{-}c;0} - \kb{\mu{-}c;1}) \otimes \Bbb{1}$, and in the
  second case pick $B = (\kb{\mu{-}b;0} - \kb{\mu{-}b;1}) \otimes \Bbb{1}$. Then the commutator
  $[A,B]$ leads again to $\pm \KB{\mu{-}b;0}{\mu{-}c;0} \otimes
  \KB{b}{c}$. 

  Therefore, one can conclude that $\mathfrak{su}_\Bbb{C}(\mathcal{H}_M^{(\mu)}
  \hat{\otimes}_k \ket{0}) \subset \mathfrak{g}_\Bbb{C}$ for all $k$. The same
  reasoning holds for  $\mathfrak{su}_\Bbb{C}(\mathcal{H}_M^{(\mu-1)}
  \hat{\otimes}_k \ket{1})$. Hence the statement follows from the previous
  lemma. 
\end{proof}

Now let us consider the control Hamiltonians $H_{\mathrm{IC},j},
H_{\mathrm{IC},M+j}$ from Equation (\ref{eq:44}). We will use  Lemma
\ref{lem:6} and an induction in $M$ to prove Thm. \ref{thm:8}, which we restate
here as a proposition.

\begin{prop} \label{prop:8}
  The dynamical group generated by the control Hamiltonians
  $H_{\mathrm{IC},j}$ with $j\in\{1,\dots,2M\}$ is identical to $\mathcal{SU}(X_M)$.
\end{prop}
 
\begin{proof}
  According to Corollary \ref{kor:1} we have to show that for each $K$, we find that
  $\mathfrak{l}_M^{[K]} = \mathfrak{su}^{[K]}(X_M)$, where
  $\mathfrak{l}_M$ denotes the Lie algebra generated by the
  $H_{\mathrm{IC},j}$ with $j\in\{1,\dots,2M\}$. Since $\mathfrak{l}_M \subset
  \mathfrak{su}(X_M)$ is trivial, only the other inclusion has to be
  shown. This will be done by induction. By Prop.~\ref{prop:7} the statement
  is true for $M=1$. Now we assume it is true for $M$ to show that it is
  true for $M{+}1$, too. To this end, consider for
  each $k\in\{1,\dots,M{+}1\}$  the Hamiltonians $H_{\mathrm{IC},j}$,
  $H_{\mathrm{IC},M+1+j}$ with $j\in\{1,\dots,M{+}1\}$ and $j\neq k$. They can be
  regarded as operators on the Hilbert space $\mathcal{H}_M$ and they generate
  a Lie algebra $\mathfrak{l}_M$ which satisfies by assumption
  \begin{equation} \label{eq:49}
    \mathfrak{l}_M^{[K]} = \mathfrak{su}^{[K]}(X_M) = \bigoplus_{\mu=1}^K
    \mathfrak{su}(\mathcal{H}_M^{(\mu)}) 
  \end{equation}
  for all $K$. As operators on $\mathcal{H}_{M+1}$, they generate the
  Lie algebra $\mathfrak{l}_M \hat{\otimes}_k \Bbb{1} \subset \mathfrak{l}_{M+1}$
  and according to (\ref{eq:49}) one finds that
    $\mathfrak{su}(\mathcal{H}_M^{(\mu)}) \hat{\otimes}_k \Bbb{1} \subset
    \mathfrak{l}_{M+1}^{[K+1]}$ holds for all $\mu \leq K$
  and
  $k\in\{1,\dots,M{+}1\}$. Thus, we can apply Lemma~\ref{lem:6} 
  and $\mathfrak{su}(\mathcal{H}_{M+1}^{(\mu)})$ is contained 
  in the Lie algebra $\mathfrak{l}_{M+1}^{[K+1]}$ for all $\mu \leq K$. But since
  $\mathfrak{l}_{M+1}^{(K)} \subset \mathfrak{su}^{[K]}(X_{M+1}) =
  \mathfrak{su}(\mathcal{H}_{M+1}^{(K)})$, one even gets
  $\mathfrak{su}^{[K]}(X_{M+1}) \subset \mathfrak{l}_{M+1}^{[K]}$, just as was
  to be shown.
\end{proof}

\subsection{Many atoms under collective control}
\label{sec:many-atoms-coll-1}

As a last topic in this section, we provide proofs for Thms.~\ref{thm:4} and 
\ref{thm:5}. To this end, recall the notation from
Sect.~\ref{sec:many-atoms-coll}. The Hilbert space is $\mathcal{H}_{\sym} =
\Bbb{C}^{M+1} \otimes \mathrm{L}^2(\Bbb{R})$ with basis
\begin{equation} \label{eq:23}
  \ket{\mu;\nu} = \ket{\nu} \otimes \ket{\mu {-} \nu}\,\text{ where }\, \nu \in\{0, \dots,
  d_\mu\}\, \text{ and }\, d_\mu= \min(\mu,M).  
\end{equation}
The charge operator is again $X_M = S_3 \otimes \unity + \unity \otimes N$
from Eq.~\eqref{eq:47} but now as an operator on $\mathcal{H}_{\sym}$ with
domain $D_{\sym}$ defined in (\ref{eq:51}) and the $\mu$-eigenspaces
$\mathcal{H}_{\sym}^{(\mu)}$ become $\mathcal{H}_{\sym}^{(\mu)} = \SP
\{\ket{\mu;\nu}\, | \, \nu \in \{0,\dots,d_\mu\}\}$; cf.\ Eq.~\eqref{eq:52}. The
control Hamiltonians are $H_{\mathrm{TC},j}$ with $j\in\{1,\dots,3\}$ defined in
\eqref{eq:31} and \eqref{eq:33}. In addition let us introduce the  operators $Y_3, Y_\pm \in
\mathfrak{su}_\Bbb{C}(X_M)$ (which denotes again the complexification of
$\mathfrak{su}(X_M)$) given by  
\begin{equation} \label{eq:43}
   Y_3 \ket{\mu;\nu} = \nu \ket{\mu;\nu}, \, Y_+^{(\mu)} = \sum_{\nu=0}^{d_\mu-1}
  \KB{\mu;\nu{+}1}{\mu;\nu}, \,  Y_-^{(\mu)} = \sum_{\nu=1}^{d_\mu}
  \KB{\mu;\nu{-}1}{\mu;\nu}.    
\end{equation}
They are related to the $H_{\mathrm{TC},j}$ by
\begin{align} 
  H_{\mathrm{TC},1} &= Y_3 - (M/2)\, \Bbb{1},\, H_{\mathrm{TC},3} = X_M
  - Y_3,\nonumber \\ 
  H_{\mathrm{TC},+} &= S_+ \otimes a = f(X_M,Y_3) Y_+,\, H_{\mathrm{TC},-}=S_- \otimes a^* = Y_-
  f(X_M,Y_3)  \label{eq:53}
\end{align}
where $f$ is a function in two variables $x,y$ given by
\begin{equation} \label{eq:59}
  f(x,y) = h_1(x,y) h_2(y) \sqrt{y},\, h_1(x,y) = \sqrt{x+1-y},\, h_2(y)
  = \sqrt{M+1-y},
\end{equation}
and $f(X_M,Y_3)$ has to be understood in the sense of functional caculus (both
operators commute). As operators on $\mathcal{H}_{\sym}^{(\mu)}$ for fixed
$\mu$, the $Y_\pm$ satisfy
\begin{equation} \label{eq:54}
  Y_+Y_- = \Bbb{1} - \kb{\mu,0},\, Y_-Y_+ = \Bbb{1} - \kb{\mu,d_\mu}
\end{equation} 
and for any function $g(y)$ which is continuous on the spectrum of $Y_3$, one finds
\begin{equation} \label{eq:55}
  Y_+ g(Y_3) = g(Y_3-\Bbb{1}) Y_+,\, Y_- g(Y_3) = g(Y_3+\Bbb{1}) Y_-. 
\end{equation}
We are now prepared for the first lemma. 

\begin{lem}\label{lem:4.13}
  The operators $H_{\mathrm{TC},1}$, $H_{\mathrm{TC},+} = S_+ \otimes a$,
  and $H_{\mathrm{TC},-} = S_- \otimes a^*$ satisfy the following commutation
  relations (as operators on $\mathcal{H}^{(\mu)}$)
   (i) $[Y_3^{n-1} H_{\mathrm{TC},+},H_{\mathrm{TC},-}] =  (X_M - Y_3) Y_3^n + (N 
    \Bbb{1} - Y_3)Y_3^n - (X_M - Y_3)(N \Bbb{1} - Y_3) \sum_{k=0}^{n-1}  \binom{n}{k} Y_3^k$ 
  and
   (ii) $[Y_3^{n+1},H_{\mathrm{TC},+}] = \sum_{k=0}^n \binom{n}{k} (-1)^{n-k} Y_3^k   H_{\mathrm{TC},+}$.
\end{lem}

\begin{proof}
  Using Eq.~\eqref{eq:53} to re-express $H_{\mathrm{TC},\pm}$ in terms of
  $Y_\pm$, $Y_3$ and $X_N$, we get for the first commutator
  \begin{equation} \label{eq:56}
    [Y_3^{n-1} H_{\mathrm{TC},+},H_{\mathrm{TC},-}] = Y_3^{n-1}
    f(X_M,Y_3)Y_+Y_f(X_M,Y_3) - Y_- f^2(X_M,Y_3)Y_3^{n-1}Y_+.
  \end{equation}
  It is easy to check that $f(X_M,Y_3) \ket{\mu;0} = 0$ holds. Together with
  (\ref{eq:54}) this leads to 
  \begin{equation} \label{eq:57}
    Y_3^{n-1}f(X_M,Y_3) Y_+Y_- = Y_3^{n-1}f(X_M,Y_3).
  \end{equation}
  With (\ref{eq:55}) we get on the other hand $Y_-f^2(X_M,Y_3)Y_+ =
  f^2(X_M,Y_3+\Bbb{1})(Y_3+\Bbb{1})^{n-1} Y_-Y_+$. Now observe that
  $h_1^2(X_M,Y_3+\Bbb{1}) h_2^2(Y_3+\Bbb{1}) \ket{\mu;d_\mu} = 0$ and use again
  (\ref{eq:54}) to get 
  \begin{equation} \label{eq:58}
    Y_-f^2(X_M,Y_3)Y_+ = f^2(X_M,Y_3+\Bbb{1})(Y_3+\Bbb{1})^{(n-1)}.
  \end{equation}
  Inserting (\ref{eq:57}) and (\ref{eq:58}) into (\ref{eq:56}) leads to
    $[Y_3^{n-1} H_{\mathrm{TC},+},H_{\mathrm{TC},-}] = Y_3^{n-1}
    f^2(X_M,Y_3)-f^2(X_M,Y_3+\Bbb{1}) (Y_3+\Bbb{1})^{n-1}$,
  where we have used the fact that $f(X_M,Y_3)$ and $Y_3$ commute. Inserting
  the definition of $f$ in (\ref{eq:59}) and expanding $(Y_3+\Bbb{1})^{n-1}$
  leads to the first commutator. 
  The second commutator follows similarly from 
    $[Y_3^{n+1},H_{\mathrm{TC},+}] = Y_3^{n+1} f(X_M,Y_3) Y_+ -
    f(X_M,Y_3)Y_+Y_3^{n+1}$ 
  and applying (\ref{eq:55}) to commute $Y_+$ to the right.
\end{proof}

We are now ready to prove Thm.~\ref{thm:4}. The statement about the dynamical
group $\mathcal{G}$ as a subgroup of $\mathcal{U}(X_M)$ is an easy consequence
of the discussion in Sect.~\ref{sec:commuting-operators}. The second
statement in Thm.~\ref{thm:4} is rephrased in the following Proposition.

\begin{prop}
  Consider the Lie algebra $\mathfrak{l}_{\mathrm{TC}} \subset
  \mathfrak{u}(X_M)$ generated by the $H_{\mathrm{TC},j}$ with $j\in\{1,\dots,3\}$ and
  $\mu \in \N$. The restriction
  $\mathfrak{l}_{\mathrm{TC}}^{(\mu)}$ of $\mathfrak{l}_{\mathrm{TC}}$ to
  $\mathcal{H}_{\sym}^{(\mu)}$ coincides with the Lie algebra
  $\mathfrak{u}(\mathcal{H}_{\sym}^{(\mu)})$ of anti-hermitian operators on
  $\mathcal{H}_{\sym}^{(\mu)}$.   
\end{prop}

\begin{proof}
  We will prove the corresponding statements for the complexifications:
  $\mathfrak{l}_{\mathrm{TC},\Bbb{C}}= \mathfrak{l}_{\mathrm{TC}} \oplus i
  \mathfrak{l}_{\mathrm{TC}} = \mathcal{B}(\mathcal{H}_{\sym}^{(\mu)})$. The
  proposition then follows from taking only anti-hermitian operators on both
  sides. Now note that $H_{\mathrm{TC},\pm} \in
  \mathfrak{l}_{\mathrm{TC},\Bbb{C}}$ since we can express them as linear
  combinations of $H_{\mathrm{TC},2}$ with the commutator of
  $H_{\mathrm{TC},1}$ and $H_{\mathrm{TC},2}$. Furthermore, $X_M$ act as $\mu
  \Bbb{1}$ on $\mathcal{H}_{\sym}^{(\mu)}$. Hence, Eq.~\eqref{eq:53} shows
  that the restriction $\mathfrak{l}_{\mathrm{TC},\Bbb{C}}^{(\mu)}$ is
  generated by $\Bbb{1}$, $Y_3$ and $H_{\mathrm{TC},\pm}$ considered as
  operators on $\mathcal{H}_{\sym}^{(\mu)}$. Note that all operators in this
  proof are operators on $\mathcal{H}_{\sym}^{(\mu)}$, and therefore we
  simplify the notation by dropping \emph{temporarily} the superscript $\mu$,
  when operators are concerned.

  The first step is to show that $Y_3^k, Y_3^jH_{\mathrm{TC},\pm} \in
  \mathfrak{l}_{\mathrm{TC},\Bbb{C}}^{(\mu)}$ holds for all $k, j \in \N_0$. This
  is done by induction. The statement is true for $k\in\{0,1\}$ and $j=0$. Now
  assume it holds for all $k\in\{0,\dots,n\}$ and $j\in\{1,\dots,n{-}1\}$. Lemma~\ref{lem:4.13}(i)
  shows that the commutator $[Y_3^{n-1} H_{\mathrm{TC},+},H_{\mathrm{TC},-}]$
  is a polynomial in $Y_3$ with $-(n+2) Y_3^{n+1}$ as leading term. Since 
  $Y_3^j \in \mathfrak{l}_{\mathrm{TC},\Bbb{C}}^{(\mu)}$ for $j\in\{0,\dots,n\}$ we can
  subtract all lower order terms and get $Y_3^{n+1} \in
  \mathfrak{l}_{\mathrm{TC}}^{(\mu)}$. To handle $Y_3^nH_{\mathrm{TC},\pm}$ we
  use Lemma~\ref{lem:4.13}(ii). The commutator $[Y_3^{n+1},H_{\mathrm{TC},+}]$ is of the
  form $P(Y_3) H_{\mathrm{TC},+}$ with an $n^{\mathrm{th}}$-order polynomial $P$. Since
  $Y_3^k H_{\mathrm{TC},+} \in \mathfrak{l}_{\mathrm{TC},\Bbb{C}}^{(\mu)}$, we can
  subtract all terms of order $k<n$ and conclude that $Y_3^n H_{\mathrm{TC},+}
  \in \mathfrak{l}_{\mathrm{TC},\Bbb{C}}^{(\mu)}$.

  Now consider a polynomial $P$ with $P(\nu) = 0$ for $\nu \neq \kappa$ and
  $P(\kappa)=1$ with $\nu,\kappa \in\{ 0, \dots, d_\mu\}$. Since all $Y_3^n$ are
  in $\mathfrak{l}_{\mathrm{TC},\Bbb{C}}^{(\mu)}$, we get $\kb{\mu;\kappa}
  = P(Y_3) \in \mathfrak{l}_{\mathrm{TC},\Bbb{C}}^{(\mu)}$. Applying the same argument
  to $Y_3^n H_{\mathrm{TC},\pm}$, we  also get
  $\KB{\mu;\kappa}{\mu;\kappa\pm 1} \in \mathfrak{l}_{\mathrm{TC},\Bbb{C}}^{(\mu)}$
  and the general case $\KB{\mu;\nu}{\mu;\lambda}$ with $\mu\neq\lambda$ can be
  treated with repeated commutators of $\KB{\mu;\kappa}{\mu;\kappa\pm 1}$ for
  different values of $\kappa$. 
\end{proof}

This proposition says that the control system with Hamiltonians
$H_{\mathrm{TC},j}$ with $j\in\{1,2,3\}$ can generate any special unitary
$U^{(\mu_0)}$ on $\mathcal{H}_{\sym}^{(\mu_0)}$ for any $\mu_0$. However,
some calculations using computer algebra, we have done for the case $M=2$
indicate that we cannot exhaust all of $\mathcal{SU}(X_M)$. In other words:
After $U^{(\mu_0)}$ is fixed, we loose the possibility to choose an
\emph{arbitrary} $U^{(\mu)} \in \mathcal{SU}(\mathcal{H}_{\sym}^{(\mu)})$ for
another $\mu$. Our analysis for two atoms suggests that the Lie algebra
generated by the $H_{\mathrm{TC},j}$ is almost as big as
$\mathfrak{su}(X_2)$, but does not contain operators of the form $A \otimes
\Bbb{1}$ with a diagonal traceless operator $A$ (except
$H_{\mathrm{TC},1}$). This observation suggests the choice of the Hamiltonians
$H_{\mathrm{CC},k}$ with $k\in\{1,\dots,M{+}1\}$ in Eq.~\eqref{eq:30}, which lead to a
dynamical group exhausting $\mathcal{SU}(X_M)$. This is shown in the next
proposition, which completes the proof of Thm.~\ref{thm:5}.

\begin{prop}
  The dynamical group generated by $H_{\mathrm{CC},k}$ with $k\in\{1,\dots,M{+}1\}$
  coincides with $\mathcal{SU}(X_M)$.
\end{prop}

\begin{proof}
  Let us introduce the operators $\kappa(k,j) \in \mathfrak{u}_\Bbb{C}(X_M)$
  (the complexification of $\mathfrak{u}(X_M)$) given by $\kappa(k,j)^{(\mu)}
  =\KB{\mu;k}{\mu;j}$ with $k,j\in\{0,\dots,M\}$ and $\kappa(k,j) = 0$ if $k \geq
  d_\mu$ and $j \leq d_\mu$, where $d_\mu = \min(\mu,M{+}1)$;
  cf.\ Eq.~\eqref{eq:23}. We can re-express $Y_\pm$ in terms of $\kappa(k,j)$ as
    $Y_+ = \sum_{k=0}^{M{-}1} \kappa(k{+}1,k)$, $Y_-=\sum_{k=1}^M \kappa(k{-}1,k)$.
  Compare this to the definition of $Y_\pm$ in (\ref{eq:43}). The truncation
  of the sums occuring for $\mu < M$ is now built into the definition of the
  $\kappa(k,j)$. Similarly we can write the $H_{\mathrm{CC},j}$ for $j\in\{1,\dots,M\}$
  as $H_{\mathrm{CC},j} = \kappa(k,k) - \kappa(k{-}1,k{-}1)$. The $\kappa(k,j)$
  are particularly useful because their commutator has the following simple
  form: $[\kappa(k,j),\kappa(p,q)] = \delta_{jp}\kappa(k,q) - \delta_{kq}
  \kappa(p,j)$. Note that all truncations for small $\mu$ are automatically
  respected. This can be used to calculate the commutator of
  $H_{\mathrm{CC},k}$ and $Y_\pm$. To this end we introduce the $M \times M$
  matrix $(A_{jk})$ with $A_{jj}=2$, $A_{j,k}=-1$ if $|j-k|=1$ and $A_{jk}=0$
  otherwise. Using $(A_{jk})$ we can write $[H_{\mathrm{CC},j},Y_+] = \sum_k A_{jk}
  \kappa(k,k{-}1)$. The matrix $(A_{jk})$ is tridiagonal, and therefore its
  determinant can be easily calculated and it equals $M{+}1$. Hence $(A_{jk})$ is
  invertible, and we can express $\kappa(j,j{-}1)$ for $j\in\{1,\dots,M\}$ as linear
  combination of the commutators $[H_{\mathrm{CC},k},Y_+]$. 

  Now consider the Lie algebra $\mathfrak{l}_{\mathrm{CC}}$  generated by
  $H_{\mathrm{CC},k}$ with $k\in\{1,\dots,M{+}1\}$ and its complexification
  $\mathfrak{l}_{\mathrm{CC},\Bbb{C}}$. We have $H_{\mathrm{TC},1} \in
  \mathfrak{l}_{\mathrm{CC}}$ since it can be written as a linear combination
  of the $H_{\mathrm{CC},j}$. In addition $H_{\mathrm{TC},3} =
  H_{\mathrm{CC},M+1} \in \mathfrak{l}_{\mathrm{CC}}$ and since $S_+ \otimes
  a$, $S_-\otimes a^*$ can be written as (complex) linear combinations of
  $H_{\mathrm{TC},3}$ and its commutator with $H_{\mathrm{TC},1}$ we get $S_+
  \otimes a, S_-\otimes a^* \in \mathfrak{l}_{\mathrm{CC},\Bbb{C}}$. To
  calculate the commutators $[H_{\mathrm{CC},j},S_+\otimes a]$ note that
  according to (\ref{eq:43}) we have $S_+ \otimes a = f(X_M,Y_3) Y_+$ and
  $f(X_M,Y_+)$ commutes with $H_{\mathrm{CC},k}$. Hence 
    $[H_{\mathrm{CC},j},S_+\otimes a] = [H_{\mathrm{CC},j}, f(X_M, Y_3) Y_+]
    = f(X_M, Y_3) [H_{\mathrm{CC},j},Y_+]  = \sum_k A_{jk} f(X_M, Y_3) \kappa(k,k{-}1)$.
  Using the reasoning from the last paragraph, we see that $f(X_M,Y_3)
  \kappa(k,k{-}1) \in \mathfrak{l}_{\mathrm{CC}}$. Similarly we can show by
  using commutators with $S_- \otimes a^*$ that all $\kappa(k,k{+}1) f(X_M,Y_3)$
  are in $\mathfrak{l}_{\mathrm{CC}}$, too. By expanding the function $f$ we
  see in this way that for $k\in\{1,...,M\}$ the operators  
\begin{equation}\label{eq:64}
    A_+ = P(k) \, \kappa(k,k{-}1),\, A_- = 
      P(k) \,  \kappa(k{-}1,k),\,
      A_3 = \kappa(k,k) - \kappa(k{-}1,k{-}1)
\end{equation}
  with $P(k):= \sqrt{X_M +(1{-}k) \Bbb{1}}  $ are elements of $\mathfrak{l}_{\mathrm{CC},\Bbb{C}}$. 

  To conclude the proof, we apply again Corollary \ref{kor:1}. Hence we have to
  consider the truncated algebra $\mathfrak{l}_{\mathrm{CC}}^{[K]}$. To this
  end, look at the subalgebra $\mathfrak{l}_{\mathrm{CC},k}$ of
  $\mathfrak{l}_{\mathrm{CC}}$ generated by the operators in
  (\ref{eq:64}). They are acting on the subspace generated by basis vectors
  $\ket{\mu;k}$, $\ket{\mu;k{-}1}$ and if we write $A_1 = A_+ + A_-$, $A_2 =
  i(A_+ - A_-)$ we get (up to an additive shift in the operator $X_M$) the
  same structure already analyzed in Lemma \ref{lem:2} (cf.\ also the operators
  $A_{\alpha,k}$ in Eq.~\eqref{eq:68}). Hence we can apply the method from
  Sect.~\ref{sec:one-atom-1} to see that for all $\mu \in\{ 0, \dots, K\}$ the
  operators $\kb{\mu;k} - \kb{\mu,k{-}1}$, $\KB{\mu;k}{\mu;k{-}1}$ and
  $\KB{\mu;k{-}1}{\mu;k}$ are elements of
  $\mathfrak{l}_{\mathrm{CC},\Bbb{C}}^{[K]}$ (provided $k \leq d_\mu$). Now we
  can generate all operators $\KB{\mu;p}{\mu,j}$ with $p,j \leq d_\mu$ by
  repeated commutators of $\KB{k}{k{-}1}$ and $\KB{k{-}1}{k}$ for different values
  of $k$. This shows that $\mathfrak{su}_{\Bbb{C}}(\mathcal{H}_{\sym}^{(\mu)})
  \subset \mathfrak{l}_{\mathrm{CC},\Bbb{C}}^{[K]}$ for all $\mu \leq K$. By
  passing to anti-selfadjoint elements we conclude that
  $\mathfrak{l}_{\mathrm{CC}}^{[K]} = \mathfrak{su}(X_M)^{[K]}$ holds for all
  $K$. Hence the statement follows from Corollary \ref{kor:1}. 
\end{proof}

\section{Strong controllability}
\label{sec:strong-contr}

The purpose of this section is to show how one can complement the block-diagonal
dynamical groups from the last section to get strong controllability. 
We add one generator which
breaks the abelian symmetry of the block-diagonal decomposition.
The proofs for pure-state controllability and strong controllability are given in Proposition~\ref{prop:9}
and Proposition~\ref{prop:10}, respectively.
This completes the proof of Theorems \ref{thm:1}, \ref{thm:7} and \ref{thm:2}.

\subsection{Pure-state controllability}

Consider a family $H_1, \dots, H_n$ of control Hamiltonians on the Hilbert
space $\mathcal{H}$ with joint domain $D \subset \mathcal{H}$ admitting a
$\mathrm{U}(1)$-symmetry defined by a charge operator $X$ with the same
domain. Since all the subspaces $\mathcal{H}^{(\mu)}$ are invariant under all
time evolutions, which can be constructed from the $H_k$,  pure-state controllability 
cannot be achieved. For rectifying this problem, we have to
add a Hamiltonian that breaks this symmetry in a specific way. We will do
so by using complementary operators as in Definition
\ref{def:1}. Hence in addition to the projections $X^{(\mu)}$, $\mu \in
\N_0$ we have the mutually orthogonal projections $E_\alpha$,
$\alpha\in\{+,0,-\}$ introduced in Sect.~\ref{sec:atoms-cavity} and the
corresponding derived structures. This includes in particular the
subprojections $X^{(\mu)}_\alpha \leq X^{(\mu)}$, $\mu \in \N_0$ and the
Hilbert spaces $\mathcal{H}^{(\mu)}_\alpha$ onto which they project. Recall,
that they satisfy $X^{(\mu)}_\alpha = E_\alpha X^{(\mu)}$ and $X^{(\mu)} =
X^{(\mu)}_- \oplus X^{(\mu)}_0 \oplus X^{(\mu)}_+$, and that for $\mu > 0$ the
$X^{(\mu)}_\pm$ are required to be non-zero. For the following discussion we
need in addition the Hilbert spaces $\mathcal{H}_{[K]} = \mathcal{H}^{[K]}
\oplus \mathcal{H}^{(K+1)}_-$, the projections $F_{[K]}$ onto them and the
group $\mathcal{SU}(X,F_{[K]})$ of $U \in \mathcal{SU}(X)$ commuting with
$F_{[K]}$. Furthermore we will indicate restrictions to the subspaces
$\mathcal{H}_{[K]}$ by a subscript $[K]$, e.g. $\mathcal{SU}_{[K]}(X,F_{[K]})$
denotes the corresponding restriction of $\mathcal{SU}(X,F_{[K]})$ which has the
form $\mathcal{SU}_{[K]}(X,F_{[K]}) = \mathcal{SU}^{[K]}(X) \oplus
\mathcal{SU}(X^{(K+1)}_-)$. Now one can prove the following lemma, which will be of
importance in the subsequent subsections. 

\begin{lem} \label{lem:7}
  Consider a strongly continuous representation $\pi: \mathrm{U}(1) \to
  \mathcal{U}(\mathcal{H})$ with charge operator $X$, an operator $H$
  complementary to $X$, and the objects just introduced.  For
  all $K \in \N$,  introduce the Lie group $\mathcal{G}_{X,F,K}$ generated by 
  $\mathcal{SU}_{[K]}(X,F_{[K]})$, $\exp(i t H)$, $t\in \Bbb{R}$ and global
  phases $\exp(i \alpha) \Bbb{1}$, $\alpha \in [0,2\pi)$. Then the group
  $\mathcal{G}_{X,F,K}$ acts transitively on the unit sphere of
  $\mathcal{H}_{[K]}$.  
\end{lem}

\begin{proof}
  Consider $\phi \in \mathcal{H}_{[K]}$ and choose $\tilde{U}_1 \in
  \mathcal{SU}_{[K]}(X,F_{[K]})$ such that $X_+^{(\mu)} \tilde{U}_1 \phi =
  0$ for all $\mu > 0$.  This is possible, since
  $\mathcal{SU}(\mathcal{H}^{(\mu)})$ acts transitively (up to a phase) on the
  unit vectors of $\mathcal{H}^{(\mu)} = \mathcal{H}_-^{(\mu)} \oplus
  \mathcal{H}_0^{(\mu)} \oplus \mathcal{H}^{(\mu)}_+$. 
  According to item (\ref{item:1}) of Def. \ref{def:1} we can find $t \in
  \Bbb{R}$ (e.g. $t=\pi/2$ will do) such that $\exp(i t H)
  \mathcal{H}_+^{(K+1)} = \mathcal{H}_-^{(K)}$ holds. Hence $\exp(i t H) \phi
  \in \mathcal{H}^{[K]}$ and we can find a $\tilde{U}_2 \in
  \mathcal{SU}_{[K]}(X,F_{[K]})$ with $\phi_1 = \tilde{U}_2 \exp(i t H)
  \tilde{U}_1 \phi \in \mathcal{H}_{[K-1]}$.  Applying this procedure $K$
  times we get $\phi_K = U_K \cdots U_1 \phi \in \mathcal{H}_{[0]}$ with $U_j
  \in \mathcal{G}_{X,F,k}$. Similarly we can find $V_1, \dots, V_K \in
  \mathcal{G}_{X,F,k}$ with $\psi_K = V_k \cdots V_1 \psi \in
  \mathcal{H}_{[0]}$. 

  Now note that the group $\mathcal{G}_{X,F,0}$ can be regarded as a subgroup
  of $\mathcal{G}_{X,F,k}$ (which acts trivially on the orthocomplement of
  $\mathcal{H}_{[0]}$ in $\mathcal{H}_{[K]}$). Hence, the statement of the
  lemma follows from the fact that, due to condition (\ref{item:2}) of
  Def. \ref{def:1}, the group $G_{X,F,0}$ acts transitively on the unit
  vectors in $\mathcal{K}_{[0]} = F_{[0]} \mathcal{H}$.   
\end{proof}

The first easy consequence of this lemma is the following result which is a
proof of Thm. \ref{thm:10} which we restate here as a proposition.

\begin{prop} \label{prop:9}
  Consider a strongly continuous representation $\pi: \mathrm{U}(1) \to
  \mathcal{U}(\mathcal{H})$ with charge operator $X$ and a family of
  selfadjoint operators $H_1, \dots, H_d$ on $\mathcal{H}$. Assume that the
  following conditions hold:
  \begin{enumerate}
  \item 
    All eigenvalues $\mu$ of $X$ are greater than or equal to $0$.
  \item 
    $H_1, \dots, H_{d-1}$ commute with $X$.
  \item 
    The dynamical group generated by $H_1, \dots, H_{d-1}$ contains
    $\mathcal{SU}(X)$.
  \item 
    The operator $H_d$ is complementary to $X$.
  \end{enumerate}
  Then the system (\ref{eq:1}) with Hamiltonians $H_0 =
  \Bbb{1}, H_1, \dots, H_d$ is pure-state controllable.
\end{prop}

\begin{proof}
  We have to show that for each pair of pure
  states $\psi, \phi \in \mathcal{H}$ and each $\epsilon>0$ there is a finite
  sequence $U_k \in \mathcal{U}(\mathcal{H})$ with $k \in\{1, \dots, N\}$ and either $U_k
  \in \mathcal{SU}(X)$, $U_k=\exp(i t H_d)$, or $U_k = \exp(i \alpha)
  \Bbb{1}$ such that $\| \psi - U_N \cdots U_1 \phi \| < \epsilon$. To this
  end, first note that we can find $K \in \N$ such that $\| \psi - F_{[K]} 
  \psi\| < \epsilon/3$ and $\|\phi - F_{[K]} \psi\| < \epsilon/3$, where
  $F_{[K]}$ is the projection defined in the first paragraph of this
  subsection.  Therefore 
  $
    \| \psi - U_N \cdots U_1 \phi \| \leq  \| \psi - F_{[K]} \psi\| + \|
    F_{[K]} \psi - U_N \cdots U_1 F_{[K]} \phi \| + \| U_N \cdots U_1 F_{[K]}
    \phi - U_N \cdots U_1 \phi \| < \epsilon
  $
  provided $\| F_{[K]} \psi - U_N \cdots U_1 F_{[K]} \phi \| <
  \epsilon/3$. Hence we can assume that $\psi, \phi \in \mathcal{H}_{[K]}$ and
  apply Lemma \ref{lem:7}. This leads to a sequence $V_1 ,\dots, V_N \in
  \mathcal{G}_{X,F,K}$ with $V_N \cdots V_1 \phi = \psi$. Now note that the
  dynamical group $\mathcal{G}$ generated by $H_0, \dots, H_{d}$ contains by
  assumption the group $\mathcal{SU}(X)$, the unitaries $\exp(i t H_{d})$
  and the global phases $\exp(i \alpha) \Bbb{1}$. Hence with the definition of
  $\mathcal{G}_{X,F,K}$, we get for $j\in\{1,\dots,N\}$ a $W_j \in \mathcal{G}$ with
  $[W_j,F_{[K]}] = 0$ and $F_{[K]} W_j = V_j$, and therefore $\psi = W_N \cdots
  W_1 \phi$. But by definition the dynamical group is the strong closure of
  monomials $U_N \cdots U_1$ with $U_j = \exp(i t_j H_{k_j})$ for some $t_j \in
  \Bbb{R}$ and $k_j \in\{ 0, \dots, d+1\}$. In other words for all $U \in
  \mathcal{G}$, $\xi \in \mathcal{H}$ and $\epsilon >0$ we can find such a
  monomial satisfying $\| U_N \cdots U_1 \xi  -U \xi\| < \epsilon$. Applying
  this statement to the operators $W_j$ and the vectors $W_{j-1} \cdots W_1
  \phi$ concludes the proof.
\end{proof}

This proposition can be applied to all systems studied in
Sect. \ref{sec:atoms-cavity}. Therefore, {\em they are all pure-state
controllable}. However, as already stated, one can even prove strong
controllability, which is the next goal.

\subsection{Approximating unitaries}
\label{sec:appr-unit}

Lemma \ref{lem:7} shows that the group $\mathcal{G}_{X,F,K}$ acts transitively
on the pure states in the Hilbert space $\mathcal{H}_{[K]}$. This implies
that there are only two possibilities for this group: either
$\mathcal{G}_{X,F,K}$ coincides with group of symplectic unitaries on
$\mathcal{H}_{[K]}$ (which is only possible if the dimension of
$\mathcal{H}_{[K]}$ is even), or it is the whole unitary group
\cite{SchiSoLea02,SchiSoLea02b,AA03}. At the same time we have seen in
Prop. \ref{prop:9} that (under appropriate 
conditions on the control Hamiltonians) each $U \in \mathcal{G}_{X,F,K}$
admits an element $W$ in the dynamical group satisfying $W \xi = U \xi$ for
all $\xi \in \mathcal{H}_{[K]}$. Proving full controllability can therefore be
reduced to two steps:
\begin{enumerate}
\item 
  Find arguments that for an infinite number of $K \in \N$, the group
  $\mathcal{G}_{X,F,K}$ cannot be unitary symplectic, such that it has to
  coincide with the full unitary group on $\mathcal{H}_{[K]}$.
\item \label{item:3}
  Show that each unitary $U \in \mathcal{U}(\mathcal{H})$ can be approximated
  by a sequence $W_K$, $K \in \N$ of unitaries of the form $W_K = U_k
  \oplus V_k$, where $U_k \in \mathcal{U}(\mathcal{H}_{[K]})$ can be chosen
  arbitrarily, while $V_K$ is a unitary on $(\Bbb{1} - F_{[K]}) \mathcal{H}$
  which is (at least partly) fixed by the choice of $U_k$. 
\end{enumerate}
The purpose of this subsection is to prove the second statement, while the
first one is postponed to Section~\ref{sec:strong-contr-2}. We start with the
following lemma:

\begin{lem} \label{lem:8}
  Consider a sequence $F_{[K]}$, $K \in \N$ of finite-rank projections
  converging strongly to $\Bbb{1}$ and satisfying $F_{[K]} \lneq
  F_{[K+1]}$. For each unitary $U \in \mathcal{U}(\mathcal{H})$ there is a
  sequence $U_{[K]}$, $K \in \N$ of partial isometries, which converges
  strongly to $U$ and satisfies $U_{[K]}^*U_{[K]} = U_{[U]} U_{[K]}^* =
  F_{[K]}$;  i.e.\ $F_{[K]}$ is the source and the target projection of $U_{[K]}$.
\end{lem}

\begin{proof}
  Let us start by introducing the space $D \subset \mathcal{H}$ of vectors
  $\xi \in \mathcal{H}$ satisfying $F_{[K]} \xi = \xi$ for a $K \in
  \N$. It is a dense subset of $\mathcal{H}$ and we can define the map
  $m: D \rightarrow \N$, $m(\xi) = \min \{ K \in \N\,|\, F_{[K]} \xi
  = \xi \}$. All operators in this proof are elements of the unit ball
  $\mathcal{B}_1(\mathcal{H}) = \{ A \in \mathcal{B}(\mathcal{H})\, | \, \|A\|
  \leq 1\}$ in $\mathcal{B}(\mathcal{H})$. A sequence $A_K$ of elements of
  $\mathcal{B}_1(\mathcal{H})$  converges to $A \in
  \mathcal{B}_1(\mathcal{H})$ iff $\lim_{K\rightarrow \infty} A_K \xi = A \xi$
  holds for all $\xi \in D$; \cite[I.3.1.2]{blackadar06}.

  Now define $A_{[K]} = F_{[K]} U F_{[K]}$. For $\xi \in D$, we have $U F_{[K]}
  \xi = U \xi$ if $K > m(\xi)$ and $\lim_{K \rightarrow \infty} F_{[K]} U \xi
  = U \xi$ since $F_{[K]}$ converges strongly to $\Bbb{1}$. Hence the strong
  limit of the $A_{[K]}$ is $U$,  similarly one can show that the strong
  limit of $A_{[K]}^*$ is $U^*$. 
  The $A_{[K]}$ are not partial isometries. We will rectify this problem by
  looking at the polar decomposition. To this end, first consider 
  $|A_{[K]}|^2 = A_{[K]}^* A_{[K]}$ and 
  $
    \| A_{[K]}^* A_{[K]} \xi - \xi \| = \| A_{[K]}^* A_{[K]} \xi - U^* U \xi
    \| \leq \| A_{[K]}^* (A_{[K]} -U) \xi\| + \| (A_{[K]}^* - U^*) U\xi\| 
   \leq \|(A_{[K]} - U) \xi\| + \|(A_{[K]}^*-U^*) U\xi\|
  $
  where we have used that $\|A_{[K]}^*\|\leq 1$ holds. Strong convergence of
  $A_{[K]}$ and $A_{[K]}^*$ implies $\lim_{K\rightarrow\infty} \| A_{[K]}^*
  A_{[K]} \xi - \xi \| =0$. Hence  $|A_{[K]}|^2$ converges strongly to
  $\Bbb{1}$. 

  The operators $A_{[K]}$ are of finite rank with support and range contained
  in $\mathcal{H}_{[K]} = F_{[K]} \mathcal{H}$. Hence the $|A_{[K]}|$  have pure point
  spectrum and their spectral decomposition is $\sum_{\lambda \in
    \sigma(|A_{[K]}|)} \lambda P_\lambda$ with eigenvalues $0 \leq \lambda
  \leq 1$ and spectral projections $P_\lambda$ satisfying $P_\lambda \leq
  F_{[K]}$ for $\lambda > 0$. Using the fact that the $P_\lambda$ are
  mutually orthogonal, we get for $|A_{[K]}|^2$:
  $
    \|  |A_{[K]}|^2 \xi - \xi\| = \| \sum_{\lambda \in \sigma(|A_{[K]}|)}
      (\lambda^2 -1) P_\lambda \phi\| = \sum_{\lambda \in
      \sigma(|A_{[K]}|)} |\lambda^2-1| \|P_\lambda \xi\| 
      = \sum_{\lambda \in \sigma(|A_{[K]}|)} |\lambda-1|(\lambda+1)
      \|P_\lambda \xi\| \geq \sum_{\lambda \in \sigma(|A_{[K]}|)} |\lambda-1|
      \|P_\lambda \xi\|
  $.
  Hence strong convergence of $|A_{[K]}|^2$ implies strong convergence of
  $|A_{[K]}|$. 

  Now we can look at the polar decomposition $A_{[K]} = W_{[K]}
  |A_{[K]}|$. The $W_{[K]}$ are partial isometries, and moreover, since support and range
  of the $A_{[K]}$ are contained in $\mathcal{K}_{[K]}$, they satisfy
  $W_{[K]}^* W_{[K]} \leq F_{[K]}$ and $W_{[K]} W_{[K]}^* \leq F_{[K]}$. In
  other words, we can look upon the $W_{[K]}$ as partial isometries on the finite
  dimensional Hilbert space $\mathcal{H}_{[K]}$. As such we can extend them
  to untaries $U_{[K]} \in \mathcal{U}(\mathcal{H}_{[K]})$ without sacrificing
  the relation to $A_{[K]}$, i.e.\ $A_{[K]} = U_{[K]} |A_{[K]}|$. As operators
  on $\mathcal{H}$, the $U_{[K]}$ are still partial isometries, but now with
  source and target projection equal to $F_{[K]}$ as stated in the lemma. 

  The only remaining point is to show that the $U_{[K]}$ converges strongly to
  $U$. This follows from
   $
    \| U_{[K]} \xi - U \xi\| \leq \| U_{[K]}\xi - A_{[K]}\xi\| + \|A_{[K]}\xi
    - U\xi \|$
   and
   $ \|U_{[K]}\xi - A_{[K]}\xi\| = \| U_{[K]} (\Bbb{1} - |A_{[K]}|)\xi\| 
  $
  since $A_{[K]}$ converges strongly to $U$ and $|A_{[K]}|$ to $\Bbb{1}$.
  \end{proof}
Now we come back to the case discussed in the beginning of this subsection
under item (\ref{item:3}):
\begin{lem} \label{lem:10}
  Consider $U$, $F_{[K]}$ and $U_{[K]}$ as in Lemma \ref{lem:8}, and an
  additional sequence of partial isometries $V_{[K]}$, $K \in \N$ with
  $V_{[K]}^*V_{[K]} = V_{[K]} V_{[K]}^* = \Bbb{1} - F_{[K]}$. The operators
  $W_{[K]} = U_{[K]} + V_{[K]}$ are unitary, and if $U$ is the strong limit of
  the $U_{[K]}$, the same is true for the $W_{[K]}$.  
\end{lem}

\begin{proof} 
  The kernels of $U_{[K]}$ and $V_{[K]}$ are $(\Bbb{1} - F_{[K]}) \mathcal{H}$
  and $\mathcal{H}_{[K]} = F_{[K]} \mathcal{H}$, respectively. These spaces are
  complementary, and therefore $W_{[K]} = U_{[K]} + V_{[K]}$ is unitary for all
  $K$. To show strong convergence, recall the space $D$ and the function
  $D  \ni \xi \mapsto m(\xi) \in \N$ introduced in the last proof. For
  $\xi \in F$ we have $W_{[K]} \xi = U_{[K]} \xi$ if $K > m(\xi)$. Hence by
  assumption $\lim_{K \rightarrow \infty} W_{[K]} \xi = \lim_{K \rightarrow
    \infty} U_{[K]} \xi = U\xi$, which implies strong convergence of
  $W_{[K]}$ to $U$.  
\end{proof}

\subsection{Strong controllability}
\label{sec:strong-contr-2}

We are now prepared to prove Theorem~\ref{thm:7}. The first
step is the following lemma announced already at the beginning
of Subsection~\ref{sec:appr-unit}.

\begin{lem} \label{lem:9}
  Consider the group $\mathcal{G}_{X,F,K}$ introduced in Lemma \ref{lem:7} and
  assume that there is a $\mu \leq K$ with $d^{(\mu)} =
  \dim(\mathcal{H}^{(\mu)}) > 2$. Then $\mathcal{G}_{X,F,K} =
  \mathcal{U}(\mathcal{H}_{[K]})$.  
\end{lem}

\begin{proof}
  Consider the group $\mathcal{SG}_{X,F,K}$ consisting of elements of
  $\mathcal{G}_{X,F,K}$ with determinant $1$. By Lemma \ref{lem:7} this group
  acts transitively on the set of pure states of the Hilbert space
  $\mathcal{H}_{[K]}$. Hence, there are only two possibilities
  left\footnote{\label{fn:1} Note that $\mathcal{H}_{[K]}$ is a finite-dimensional Hilbert
    space. Hence after fixing a basis $e_1, \dots, e_d$ it can be identified with
    $\Bbb{C}^d$.}: $\mathcal{SG}_{X,F,K}$ coincides
  either with the unitary symplectic group $\mathrm{USp}(\mathcal{H}_{[K]})$
  or with the full unitary group $\mathcal{U}(\mathcal{H}_{[K]})$; cf.\
  \cite{SchiSoLea02,SchiSoLea02b,AA03}. Assume $\mathcal{SG}_{X,F,K} =
  \mathrm{USp}(\mathcal{K}_{[K]})$ holds. This would imply that
  $\mathcal{SG}_{X,F,K}$ is self-conjugate (or more precisely the
  representation given by the identity map on $\mathcal{SG}_{X,F,K} \subset
  \mathcal{B}(\mathcal{H}_{[K]})$ is self-conjugate). In other words, there
  would be a unitary $V \in \mathcal{U}(\mathcal{H}_{[K]})$ with $V U V^* =
  \bar{U}$ for all $U \in \mathcal{SG}_{X,F,K}$. Here $\bar{U}$ denotes
  complex conjugation  in an arbitrary but fixed basis  (cf.\ footnote~\ref{fn:1}). 

  Now consider $\mathcal{SU}(\mathcal{H}^{(\mu)})$ with $d^{(\mu)} > 2$. It
  can be identified with $\mathrm{SU}(d)$ in its first fundamental
  representation $\lambda_1$ (i.e.\ the ``defining'' representation). At the
  same time it is a subgroup of $\mathcal{SG}_{X,F,K}$ (one which acts
  nontrivially only on $\mathcal{H}^{(\mu)} \subset
  \mathcal{H}_{[K]}$). Existence of a $V$ as in the last paragraph would imply
  that $\lambda_1$ is unitarily equivalent to its conjugate representation,
  which is the $d-1^{\mathrm{st}}$ fundamental representation. This is
  impossible if $d^{(\mu)} > 2$ holds. Hence $V$ with the described properties
  does not exist and $\mathcal{SG}_{X,F,k}$ has to coincide with
  $\mathcal{SU}(\mathcal{H}_{[K]})$ and therefore $\mathcal{G}_{X,F,K} =
  \mathcal{U}(\mathcal{H}_{[K]})$ as stated.
\end{proof}

Finally we can conclude the proof of Thm. \ref{thm:7} which we restate here as
the following proposition:

\begin{prop} \label{prop:10}
  A control system (\ref{eq:1}) with control Hamiltonians
  $H_0=\Bbb{1}, \dots, H_d$ satisfying the conditions from
  Prop.~\ref{prop:9} is strongly controllable, if $d^{(\mu)} = \dim
  \mathcal{H}^{(\mu)} > 2$ for at least one $\mu \in \N$. 
\end{prop}

\begin{proof}
  Consider an arbitrary unitary $U \in \mathcal{U}(\mathcal{H})$. By Lemma
  \ref{lem:8}, there is a sequence of partial isometries $U_{[K]}$ converging
  strongly to $U$, and by Lemma \ref{lem:9} we can assume that $U_{[K]} \in
  \mathcal{G}_{X,F,K}$. Now considering the dynamical group $\mathcal{G}$ generated by
  the $H_j$, define the subgroup $\mathcal{G}(F_{[K]})$ of $U \in
  \mathcal{G}$ commuting with $F_{[K]}$, and the restriction
  $\mathcal{G}_{[K]}$ of $\mathcal{G}(F_{[K]})$ to $\mathcal{H}_ {[K]}$. The assumptions
  on the $H_j$ imply that $\mathcal{G}_{[K]} = \mathcal{G}_{X,F,K} =
  \mathcal{U}(\mathcal{H}_{[K]})$. Hence there is a sequence $W_K$, $K \in
  \N$ of unitaries with $W_{[K]} \in \mathcal{G}(F_{[K]}) \subset
  \mathcal{G}$ and $F_{[K]} W_{[K]} = U_{[K]}$. Since $U_{[K]}$ converges to  $U$
  strongly, Lemma \ref{lem:10} implies that the strong limit of the $W_{[K]}$
  is $U$, which was to show.   
\end{proof}

This proposition shows strong controllability for all the systems studied in
Sect.~\ref{sec:atoms-cavity}. The only exception is one atom interacting with
one harmonic oscillator (Sect. \ref{sec:one-atom}). Here we have $d^{(\mu)} =
\dim \mathcal{H}^{(\mu)} \leq 2$ and we \emph{can} actually find a unitary $V$
with $V U V = \bar{U}$ for all $U \in
\mathcal{SU}_{[K]}(X_1,F_{[K]})$. However, the elements
$U$ of $\mathcal{SU}(X_1)$ are block diagonal where the blocks $U^{(\mu)} \in
\mathcal{SU}(\mathcal{H}^{(\mu)})$ can be chosen independently. This implies $V
\in \mathcal{SU}_{K]}(X_1,F_{[K]})$, which is incompatible with $V
H_{\mathrm{JC},4} V^* = - H_{\mathrm{JC},4}$ (cf.\ Eq.~\eqref{eq:25} for
the definition of $\mathcal{H}_{\mathrm{JC},4}$) which would be necessary for the
group $\mathcal{G}_{X_1,F,K}$ to be self-conjugate. Hence we can proceed as in
the proof of Prop.~\ref{prop:10} to prove Thm.~\ref{thm:1}.  

\section{Conclusions and Outlook}

Many of the difficulties of quantum control theory in infinite dimensions
arise from the fact that, due to unbounded operators, the
group $\mathcal{U}(\mathcal{H})$ of all unitaries on an infinite-dimensional
separable Hilbert space $\mathcal{H}$ is in fact no Lie group as long as it is
equipped with the strong topology, which inevitably is the correct choice when studying
questions of quantum dynamics. Yet $\mathcal{U}(\mathcal{H})$ contains
\begin{table}[Ht!]
  \centering
\caption{Controllability results for several 2-level atoms in a cavity as derived here.}
  \begin{tabular}{l@{\hspace{2mm}}l@{\hspace{2mm}}ll}
    \hline \\[-4.8mm]\hline\\[-4.8mm]
    System & Control Hamiltonians & \multicolumn{2}{c}{------------ Controllability ------------}\\[-1mm]
    & & \multicolumn{2}{l}{system algebra\ $\fg$, dynamic\ group $\mathcal{G}$}\\
    \hline\\[-4.8mm]\hline\\[-4mm]
    one atom & $H_{\mathrm{JC},j}$, $j=1,2$, Eq.~\eqref{eq:9} & $\fg=\su(X_1)$,
    $\mathcal{G}=\mathcal{SU}(X_1)$ &[Thm.~\ref{thm:3}] \\[-1mm]
	& \multicolumn{3}{c}{\xhrulefill{black}{.5pt}}\\[-0mm] 
    & $H_{\mathrm{JC},j}$, $j=1,2$, Eq.~\eqref{eq:9} & strongly controllable$^{a}$ & \\
    & $H_{\mathrm{JC},4}$, Eq.~\ref{eq:25} & \quad with $\mathcal{G}=\mathcal{U}(\mathcal{H})$ & [Thm.~\ref{thm:1}]\\
    \hline\\[-5.5mm]\hline\\[-4mm]
    $M$ atoms & $H_{\mathrm{IC},j}$, $j=1,\dots 2M$ &$\fg=\su(X_M)$ and &\\
    \multicolumn{2}{l}{\quad with individual controls of Eq.~\eqref{eq:44} } & 
    $\mathcal{G} = \mathcal{SU}(X_M)$ & [Thm.~\ref{thm:8}]\\[-1mm]
	& \multicolumn{3}{c}{\xhrulefill{black}{.5pt}}\\[-0mm]
    & $H_{\mathrm{IC},j}$, $j=1,\dots 2M+1$ & strongly controllable$^{a}$ &\\
    \multicolumn{2}{l}{\quad with individual controls of Eqs.~(\ref{eq:44},\ref{eq:45})} & \quad with $\mathcal{G}=\mathcal{U}(\mathcal{H})$ & [Thm.~\ref{thm:7}]\\
    \hline\\[-5.5mm]\hline\\[-3mm]
    $M$ atoms & $H_{\mathrm{TC},j}$, $j=1,2,3$ & $\fg \subset \uu(X_M)$ and &\\
    \multicolumn{2}{l}{\quad under collective control of Eq.~\eqref{eq:31}} &  $\mathcal{G} \subset \mathcal{U}(X_M)$ & [Thm.~\ref{thm:4}]\\[-1mm]
	& \multicolumn{3}{c}{\xhrulefill{black}{.5pt}}\\[-0mm]
    & $H_{\mathrm{CC},j}$, $j=1,\dots,M+1$ & $\fg=\su(X_M)$ and &\\
     \multicolumn{2}{l}{\quad under collective control of  Eq.~\eqref{eq:30}} &  $\mathcal{G} = \mathcal{SU}(X_M)$ & [Thm.~\ref{thm:5}]\\[-1mm]
	& \multicolumn{3}{c}{\xhrulefill{black}{.5pt}}\\[-0mm]
    & $H_{\mathrm{CC},j}$, $j=1,\dots,M+2$ & strongly controllable$^{a}$ &\\
     \multicolumn{2}{l}{\quad under collective control of Eq.~\eqref{eq:30}} &  \quad with $\mathcal{G}=\mathcal{U}(\mathcal{H})$ & [Thm.~\ref{thm:2}]
    \\
    \hline\\[-4.8mm]\hline
  \multicolumn{4}{l}{$^{a}$\footnotesize{Here in the strong topology,  no system algebra or exponential map exists.}}
  \end{tabular}
  \label{tab:1}
\end{table}
a plethora of subgroups which are still infinite-dimensional while admitting a
proper Lie structure -- including in particular a Lie algebra~$\mathfrak{l}$ consisting of
unbounded operators and a well-defined exponential map. An important 
example are those unitaries with an abelian $\mathrm{U}(1)$-symmetry,
which in the Jaynes-Cummings model relates to a kind of particle-number operator.

As shown here, this  infinite-dimensional
system Lie algebra $\mathfrak{l}$  can be exploited for control theory
in infinite dimensions in close analogy to the finite-dimensional case. 
Due to the in-born symmetry of $\mathfrak{l}$ and
the corresponding Lie group $\mathcal{G}$, full controllability cannot be achieved
that way. Yet we have also shown that this problem can readily be overcome by 
complementary methods directly on the group level.

For several
2-level atoms interacting with one harmonic oscillator (e.g.,\ a cavity mode or a phonon mode), 
these methods allowed us to extend previous results substantially, in particular 
in two aspects also summarized in Table \ref{tab:1}:
(A) We have answered approximate control and convergence questions for
asymptotically vanishing control error. (B) Our results include not only
reachability of states, but also its operator lift, i.e.\ simulability of unitary gates. 
To this end, we have introduced the notion of \emph{strong controllability}, and we
have shown that all systems under consideration require only a fairly small set of
control Hamiltonians for guaranteeing strong controllability,  i.e.\ simulability. 
--- Thus we anticipate the methods introduced here will
find wide application to systematically characterize experimental set-ups of 
cavity QED and ion-traps in terms of pure-state controllability and simulability.

 \section*{Acknowledgements}
\footnotesize{
This work was supported in part by the  {\sc eu}
through the integrated programmes \mbox{{\sc q-essence}} and {\sc siqs},  
and the {\sc eu-strep} {\sc coquit}, and moreover
by the Bavarian Excellence Network {\sc enb}
via the international doctorate programme of excellence
{\em Quantum Computing, Control, and Communication} ({\sc qccc}),
by {\em Deutsche Forschungsgemeinschaft} ({\sc dfg}) in the
collaborative research centre {\sc sfb}~631 as well as the international 
research group {\sc for} 1482 through the grant {\sc schu}~1374/2-1.
}

\section*{References}

\providecommand{\newblock}{}

\clearpage

\end{document}